\documentclass[11pt,letterpaper]{article}

\usepackage{enumerate}
\usepackage[shortlabels]{enumitem}
\usepackage{geometry}
\usepackage{microtype}
\usepackage{url}
\usepackage{amsfonts}
\usepackage{graphicx}
\usepackage{subcaption}
\usepackage{amsthm}
\usepackage{amssymb}
\usepackage{amsmath}
\usepackage{pgfplots}
\usepackage{tikz}
\usetikzlibrary{
  graphs,
  graphs.standard
}

\pgfplotsset{compat=1.18}

\usepackage[ruled,vlined,linesnumbered]{algorithm2e}

\usepackage{nicefrac}
\DontPrintSemicolon

\let\oldnl\nl
\newcommand{\nonl}{\renewcommand{\nl}{\let\nl\oldnl}}
\newcommand\numberthis{\addtocounter{equation}{1}\tag{\theequation}}

\SetCommentSty{mycommfont}

\usepackage{comment}

\LinesNumbered 
\SetAlgoSkip{smallskip}
\usepackage{hyperref}
\def\tmp#1#2#3{%
  \definecolor{Hy#1color}{#2}{#3}%
  \hypersetup{#1color=Hy#1color}}
\tmp{link}{HTML}{800006}
\tmp{cite}{HTML}{2E7E2A}
\tmp{file}{HTML}{131877}
\tmp{url} {HTML}{8A0087}
\tmp{menu}{HTML}{727500}
\tmp{run} {HTML}{137776}
\def\tmp#1#2{%
  \colorlet{Hy#1bordercolor}{Hy#1color#2}%
  \hypersetup{colorlinks, #1bordercolor=Hy#1bordercolor}}
\tmp{link}{!60!white}
\tmp{cite}{!60!white}
\tmp{file}{!60!white}
\tmp{url} {!60!white}
\tmp{menu}{!60!white}
\tmp{run} {!60!white}

\usepackage{natbib}

\usepackage{enumitem}

\usepackage{cleveref}

\usepackage{thmtools}

\newtheorem{theorem}{Theorem}[section]
\newtheorem{lemma}[theorem]{Lemma}

\newtheorem{corollary}[theorem]{Corollary}
\theoremstyle{definition}
\newtheorem{definition}[theorem]{Definition}
\newtheorem{claim}[theorem]{Claim}%

\usepackage{bm}
\renewcommand{\vec}[1]{\boldsymbol{#1}}

\DeclareMathOperator*{\argmax}{arg\,max}

\newcommand{\Instance}{\ensuremath{\mathcal{I}}}

\newcommand{\InstanceProphet}{\ensuremath{\mathcal{I}_{\textsc{prophet}}}}

\newcommand{\InstanceConstraint}{\ensuremath{\mathcal{\hat{I}}
}}

\newcommand{\EX}{\ensuremath{\mathbb{E}}}

\newcommand{\FixedCostStrategy}{\ensuremath{\pi_{\textsc{fixed}}}\xspace}

\newcommand{\MainStrategy}{\ensuremath{\pi_{\textsc{main}}}\xspace}
\newcommand{\InstanceInterval}
{\ensuremath{\mathcal{I}_{\textsc{sl}}}}
\usepackage[tt=false]{libertine}

\allowdisplaybreaks

\makeatletter
\newenvironment{strategy}[1][htb]{%
   \begin{algorithm}[#1]%
  }{\end{algorithm}}
\makeatother

\usepackage{authblk}

\title{Pandora's Box Problem With Time Constraints}

\author[1,2]{Georgios Amanatidis}
\author[3]{Ben Berger}
\author[4]{Tomer Ezra}
\author[5]{Michal Feldman}
\author[6]{\\Federico Fusco}
\author[7]{Rebecca Reiffenh{\"a}user}
\author[8]{Artem Tsikiridis}

\affil[1]{Athens University of Economics and Business, Athens, Greece}
\affil[2]{Archimedes, Athena Research Center, Athens, Greece}
\affil[3]{Offchain Labs, Inc., Clifton, NJ, USA}
\affil[4]{Harvard University, Cambridge, MA, USA}
\affil[5]{Tel Aviv University, Tel Aviv, Israel}
\affil[6]{Sapienza University of Rome, Rome, Italy}
\affil[7]{University of Amsterdam, Amsterdam, The Netherlands}
\affil[8]{Technical University of Munich, Munich, Germany}

\date{} 

\begin{document}

\maketitle
\begin{abstract}
\noindent  The Pandora's Box problem models the search for the best alternative when evaluation is costly. In the simplest variant, a decision maker is presented with $n$ boxes, each associated with a cost of inspection and a hidden random reward. The decision maker inspects a subset of these boxes one after the other, in a possibly adaptive order, and gains the difference between the largest revealed reward and the sum of the inspection costs. 
    Although this classic version is well understood (Weitzman 1979), there is a flourishing recent literature on variants of the problem. Here we introduce a general framework---the Pandora's Box Over Time problem---that captures a wide range of variants where time plays a role, e.g., by constraining the schedules of exploration and influencing costs and rewards.
    In our framework, boxes have time-dependent rewards and costs, whereas inspection may require a box-specific processing time. Moreover, once a box is inspected, its reward may deteriorate over time. Our main result is an efficient constant-factor approximation to the optimal strategy for the Pandora's Box Over Time problem, which is generally NP-hard to compute. We further obtain improved results for the natural special cases where boxes have no processing time, boxes are available only in specific time slots, or when costs and reward distributions are time-independent (but rewards may still deteriorate after inspection). 
\end{abstract}

\section{Introduction}\label{sec:intro}

    In the classic version of the Pandora's Box problem, introduced in the seminal paper by \citet{Weitzman79}, a decision maker is presented with $n$ boxes to explore, each characterized by an inspection cost and a hidden random reward. 
    A (possibly adaptive) strategy consists of sequentially inspecting the boxes and then picking the largest observed reward, for a net utility given by the largest reward observed minus the sum of the incurred inspection costs.  Surprisingly, Weitzman showed that this complex problem can be solved efficiently via a simple greedy procedure. 
    
    The Pandora's Box problem naturally models situations like hiring, buying a house, or choosing a school: the decision maker has some prior information on each alternative, but only a costly inspection (e.g., interviewing a candidate, visiting a house, or a school) reveals their actual values. Motivated by these applications, in recent years, many variants of the original problem have been considered (we refer the interested reader to the survey by \citet{BeyhaghiCai23}, and the Related Work). In this work, we study the impact of time on the problem by proposing a unified framework that captures and generalizes existing models. Notably, our approach addresses open directions raised in seminal works. For example, \citet{Weitzman79} explicitly calls to study variants with a ``binding time horizon'' and the work of \citet{Olszewski15} highlights the need to ``address a version of Weitzman’s problem in which the prizes (offers) do not remain permanently available.''
    
    Time can impact in many different ways, especially given the sequential nature of the Pandora's Box problem. For instance, a company may have a strict deadline for hiring a new worker, so only a limited number of candidates can be interviewed. Moreover, good candidates may only remain on the market for a short period or expect a definite answer within a certain amount of time after the interview. 
    Conversely, in the housing market, it is common that houses that remain too long on the market decrease their price (i.e., their cost) so that the corresponding reward increases, or that some days are more convenient to arrange a visit (thus reducing the cost of inspection).
    As a further example, when picking a yearly gym membership or sampling perishable goods like food, selecting an option only some time after testing it may reduce the leftover value one can extract. 
    
We introduce a general framework, the \textit{Pandora's Box Over Time} problem, which captures these phenomena. Here, each box is associated with a processing time,  time-dependent inspection cost, and random reward, and with a discounting function that specifies how the value of the realized reward deteriorates in the interval of time between inspection and selection. 
Our model captures all of these aspects, and, somewhat surprisingly, it still allows for good approximating strategies, despite its generality. However,
it is clear from the above examples that many applications do not exhibit all types of time dependency at once. Therefore, we also focus on natural restrictions, for which we show that approximation guarantees can be significantly improved.

\smallskip
\noindent\textbf{Our Contribution.}
        Beyond providing a general framework for Pandora's Box Over Time problem, we provide the following main result: 
        \begin{itemize}
            \item We show how to efficiently compute a $21.3$-approximation to the optimal strategy for Pandora's Box Over Time (\Cref{thm:general_result}). As our model captures the Free-Order Prophet Inequality problem \citep{AgrawalSZ20}, we also observe that finding the optimal solution is generally NP-hard.
        \end{itemize}

        We further investigate three special cases of independent interest, where improved approximation factors can be obtained.
        \begin{itemize}
            \item When processing times are all zero (\textit{Pandora's Box Over Time With Instant Inspection problem}), but all other parameters may vary, we obtain a $(4+\varepsilon)$-approximation (\Cref{theorem:instant-proc-main}).
            \item To illustrate the versatility of our framework, we study the \textit{Pandora's Box With Time Slots} problem, for which we get an improved $(4+\varepsilon)$-approximation (\Cref{theorem:time-slots-main}). Here, the boxes have time-invariant cost and value, and no processing time, but can only be explored in box-dependent time intervals. Even though our model \textit{does not have ``hardwired'' feasibility constraints} on the exploration, its generality allows us to capture them easily. 
            \item Finally, we consider the situation where values of inspected boxes may degrade over time, but both the costs and the distributions of the rewards are not time-dependent (\textit{Pandora's Box With Value Discounting} problem). We obtain a $1.37$-approximation for this variant (\Cref{theorem:discounting-main}).
        \end{itemize}

\smallskip
\noindent    Our framework is general enough to capture, and provide results for, several models already studied in the literature.
    \begin{itemize}
        \item The original paper introducing the optimal solution to the Pandora's Box problem by \citet{Weitzman79} also encompasses a specific discounting factor for costs and rewards, as well as a possibly non-uniform processing time. These are captured as special cases of time-dependent costs and rewards, as well as value discounting functions (which regulate the deterioration of inspected rewards) in our general model.
        \item In the \textit{Pandora's Box With Commitment} problem (introduced as an open problem for the first time in the extended version of \citet{Olszewski15} and then studied by \citet{FuLX18}, and \citet{Segev021}), boxes can be inspected in any order, but the reward can only be collected immediately upon inspection. We capture this model with a value discounting function that drops to zero after the processing time has elapsed, so that a box retains value only at the round it is inspected. 
      
        \item Our model also captures ``offline'' constraints that do not seem time-related at first, such as combinatorial constraints on the boxes that can be inspected \citep{Singla18}. Using processing times and time-varying cost functions, we capture both cardinality and knapsack constraints. Note that the idea of using value discounting with respect to time, thus allowing the flexibility to address these types of constraints, was already suggested by \citet{Weitzman79}, albeit in a more restricted way.
        \item Finally, our model strictly generalizes the \textit{Online Pandora's Box} problem \citep{EsfandiariHLM19}, where the order in which the boxes must be processed is fixed in advance. This can be achieved in our model by setting the cost function to a sufficiently large value for all time steps except the one proposed by the online ordering.
    \end{itemize}

\smallskip
\noindent\textbf{Technical Challenges and Techniques.}
    The crucial difficulty at the heart of the Pandora's Box problem lies in finding the right balance between exploration and exploitation: inspecting new boxes may improve the reward, but it is more costly. This trade-off is solved by \citet{Weitzman79} via the notion of \textit{reservation value}, an index of the intrinsic value of each box. However, such a technical tool is not robust to modifications in the model, as the introduction of time-varying parameters implies that the importance of a box crucially depends on the time step at which it is inspected and, possibly, chosen. 
    
    The first step of our approach consists of creating a copy for each box $b_i$ for each time step $t$, so that box $b_i^t$ is a proxy for the strategy inspecting box $i$ at time $t$. This procedure allows us to associate a fixed cost and a random reward with each box. Still, it creates three problems: (i) we want to avoid that two proxies of the same boxes are inspected together, (ii) we want to force the strategy to inspect a box $b_i^t$ exactly at time $t$, and (iii) we want to make sure that each inspected box has enough time to be processed. 
    
    As a second step in our construction, we impose compatibility constraints on the boxes our strategy inspects to enforce (i) to (iii). 
    We do so by combining a submodular maximization routine with an adaptivity gap result \citep{BradacSZ19}, as in the work of \citet{Singla18}. In particular, the constraints induced by our model are captured by matching constraints in a block bipartite hypergraph, for which we derive an approximation algorithm. Matching constraints naturally capture the notion of complementarity between boxes (i), while the block bipartite structure ensures that enough time is associated with each box (ii)-(iii). From a technical point of view, we provide the first constant factor approximation algorithm for submodular maximization with Block Bipartite Hypergraphs (\emph{Submodular Block Matching}), which is based on contention resolution schemes \citep{Feldman13}.
    
    Finally, our last step exploits a reduction to a prophet inequality (see also \citet{EsfandiariHLM19}) to find a good stopping rule. Notably, since we are reducing to a problem where options are immediately discarded or accepted, we can ignore the deterioration of the reward in inspected boxes.  

    For the Pandora's Box With Time Slots and the Pandora's Box Over Time With Instant Inspection problems, we exploit the simplified structure of the hypergraph constraint (which collapses to simpler constraints) to achieve a better approximation result. Finally, for the Pandora's Box With Value Discounting problem, we exploit a reduction to a free-order prophet inequality, which directly provides a better approximation guarantee \citep{hill83,BubnaC23}.

\medskip
\noindent\textbf{Further Related Work.}
\label{sec:related}
    Many other versions of Pandora's Box problem have been investigated in recent years, including problems with non-obligatory inspection \citep{Doval18,BeyhaghiK19,BeyhaghiC22,FuJD22}, with interdependent valuations \citep{ChawlaGTTZ20,ChawlaGMT21}, combinatorial costs \citep{BergerEFF23}, precedence constraints \citep{BoodaghiansFLL20}, matching \citep{BowersW24}, and contexts \citep{AtsidakouCGPT24}.
    The Pandora's Box problem has also been investigated from a contract design perspective \citep{ezra2024sequential,HoeferSS25}, and in the learning setting \citep{Guo0T021,GergatsouliT22,GatmiryKSW24,HeuserK25}.
    
\medskip
\noindent\textbf{Comparison with Conference Versions.} {This work merges and extends the conference versions of the works of \citet{BergerEFF24} and \citet{AmanatidisFRT24}. \citet{BergerEFF24} introduced the Pandora's Box With Time Slots problem, discussed in Section~\ref{subsec:timeslots}, while \citet{AmanatidisFRT24} presented the more general Pandora's Box Over Time problem (see Section \ref{sec:general}), including two special cases. In Section~\ref{subsec:timeslots}, we first show that the Pandora's Box With Time Slots problem is, essentially, equivalent to a special case of the Pandora's Box  Over Time With Instant Inspection problem (also studied in \citet{AmanatidisFRT24} and presented in Section~\ref{subsec:instant} of this work). Moreover, we obtain an improved approximation guarantee of $(4+\varepsilon)$ for the Pandora's Box Over Time With Instant Inspection problem, which, naturally, also applies to the version with Time Slots. This result improves upon both the $(8+\varepsilon)$-approximation of \citet{AmanatidisFRT24} for the variant with instant inspections and the $\nicefrac{4e}{(e-1)}\approx 6.3$-approximation of \citet{BergerEFF24} for the variant with time slots. For more details, we refer the reader to Sections~\ref{subsec:instant} and~\ref{subsec:timeslots}.

\section{Model and Preliminaries}\label{sec:prelims}

We study a generalization of the Pandora's Box problem, in which time plays a crucial role in the decision-making process. We call this variant the \emph{Pandora's Box Over Time} problem, or simply {Pandora's Box Over Time}. There is a set of boxes $[n] = \left\{1, \dots, n\right\}$. A \emph{strategy} $\pi$ inspects boxes sequentially and (possibly) adaptively. To be more precise, at each round $t$ (starting from $t=1$), the strategy $\pi$ can take the following actions: (a) inspect a box $i \in [n]$ that has not been considered before (if possible) and proceed to round $t+1$, (b) stay idle\footnote{One way a strategy $\pi$ could simulate option (b) using only actions of the classic Pandora's Box problem is by inspecting a ``dummy'' box that can be processed instantly and has a deterministic reward and cost of $0$. It is without loss of generality to assume that such boxes are readily available.} at round $t$ and proceed to round $t+1$, or (c) \emph{halt}. Given a strategy $\pi$, we denote by $T_{\pi}$ the random round it halts and by $S(\pi)$ the random tuple of boxes that have been inspected by round $T_{\pi}$, ordered by their respective time of inspection.

Unlike the classic Pandora's Box problem, where each box $i \in [n]$ is associated with a scalar cost, here $i$ is associated with a \emph{cost function} $\bar{c}_i: \mathbb{Z}_{>0} \to \mathbb{R}_{\geq 0}$. This extension allows us to model the changing cost of a box based on the time it is inspected. In general, we do not impose any assumptions on the form of the cost functions. Additionally, each box has a \emph{processing time} $p_i \in \mathbb{N}$, representing the number of rounds a strategy must wait before inspecting another box. Specifically, if a strategy $\pi$ chooses to inspect box $i \in [n]$ at round $t$ (in the sense that the inspection begins at round $t$), then $\pi$ may inspect the next box starting from round $t' = t + 1 + p_i$.

Our model also captures scenarios in which both the value sampled from each inspected box and the final claimed value depend on time. Specifically, the reward of each box $i \in [n]$, the inspection of which begins at time $t$, is drawn from a probability distribution $D_{it}$. We denote by $V_{it}$ the random variable representing the reward drawn from $D_{it}$. We assume that all distributions $D_{it}$ for $i \in [n]$ and $t \in \mathbb{Z}_{>0}$ are independent. 
Additionally, the reward in each box $i \in [n]$ may degrade over time from the moment 
its inspection is completed to the moment the box is (potentially) chosen. Formally, each box is associated with a \emph{value discounting} function $\bar{v}_i : \mathbb{R}_{\geq 0} \times \mathbb{N} \to \mathbb{R}_{\geq 0}$, known to the decision maker. Consider a strategy $\pi$ that halts at time $T_{\pi}$. 
Suppose $\pi$ inspects box $i \in [n]$ at time $t \leq T_{\pi}$ and samples reward $V_{it}$ from $D_{it}$. If the strategy chooses to collect the reward from box $i$ after $\tau=T_{\pi}-t -p_i$ rounds, it gets a reward of $\bar{v}_i(V_{it}, \tau) \in \mathbb{R}_{\geq 0}$ rather than $V_{it}$. Note that we assume $\bar{v}_i(V_{it}, \tau) \leq \bar{v}_i(V_{it}, 0) = V_{it}$ for all $\tau \in \mathbb{N}$. For instance, $\bar{v}_i$ could be non-increasing in the variable measuring the time passed since inspection, although our model captures more complex behaviors.


An instance of  Pandora's Box Over Time is \(
    \Instance = \big(\bar{c}_i, p_i, (V_{it})_{t \in [H]}, \bar{v}_i\big)_{i \in [n]},
\) where $H$ (which is at least  $n + \sum_{i=1}^n p_i$) denotes the \emph{time horizon} of the instance. That is, there are $nH$ reward probability distributions to which access is given as part of the input, and after round $H$ all the rewards are assumed to be $0$. 
Given a strategy $\pi$ for  $\mathcal{I}$, we use $t_i(\pi) \in \{1, \dots, T_{\pi}-p_i\}$ to denote the random round at which strategy $\pi$ starts inspecting box $i \in S(\pi)$.\footnote{Here we implicitly exclude any strategy that halts while a box is being inspected. It is easy to see that these are irrelevant indeed, as any such strategy is dominated by the strategy that differs only in that it avoids opening the very last box whenever its inspection is not to be completed.} Furthermore, we define the random utility of the decision maker for strategy $\pi$ as
\begin{equation}\label{eq:utility}
    u_{\Instance}(\pi) := \max_{i \in S(\pi)} \bar{v}_i\left(V_{it_i(\pi)}, T_{\pi} - t_i(\pi)-p_i\right) - \sum_{i \in S(\pi)} \bar{c}_i\left(t_i(\pi)\right)
\end{equation}
i.e., we assume that when a strategy halts, \textit{the best available reward is always collected} at time $T_{\pi}$.
We use $\pi^*$ to define an optimal strategy, i.e., $\pi^* \in \argmax_{\pi} \EX
[u_{\Instance}(\pi)]$, and we say that a strategy $\pi$ is an $\alpha$-approximation of an optimal strategy if $\alpha \cdot \EX
[u_{\Instance}(\pi)] \geq \EX
[u_{\Instance}(\pi^*)]$, for $\alpha \geq 1$.

Finally, it is not hard to observe that both the classic Pandora's Box problem and its version with Commitment (where, recall, only the last reward can be collected) are special cases of our problem. To see this, consider an instance $\left(c_i, V_i\right)_{i \in [n]}$ of Pandora's Box. We can construct an instance $\Instance$ of Pandora's Box Over Time such that for each $i \in [n]$, $p_i = 0$, $\bar{c}_i(t) = c_i$, and $V_{it}=V_i$ for all $t \in \mathbb{Z}_{>0}$. Then, by setting $\bar{v}_i(V_i, t) = V_i$ for all $t \in \mathbb{N}$ (respectively, $\bar{v}_i(V_i, 0) = V_i$ and $\bar{v}_i(\cdot, t) = 0$ for all $t > 0$), the utility obtained by a strategy $\pi$ for $\left(c_i, V_i\right)_{i \in [n]}$ in Pandora's Box (respectively, Pandora's Box With Commitment) problem coincides with \eqref{eq:utility} for $(\Instance, \pi)$.

\subsection{A Class of Related Stochastic Optimization Problems}\label{sec:constrained}

In this section, we present a related class of Pandora's Box problems with added constraints on the sequences of inspected boxes. We call this variant the \emph{Constrained Pandora's Box} problem or just {Constrained Pandora's Box}. This class of stochastic optimization problems, ever since being proposed by \citet{Singla18}, has sparked a rich line of work (see Further Related Work in \Cref{sec:related}). Our purpose is to relate Pandora's Box Over Time to instances of Constrained Pandora's Box, to leverage known results.
In this scenario, as in Pandora's Box, there is a set $[n] = \{1, \dots, n\}$ of boxes, with each box $i \in [n]$ containing a random variable $V_i$ drawn from a publicly known, non-negative distribution $D_i$. The distributions $D_1, \dots, D_n$ are independent. 
Moreover, each box $i \in [n]$ is associated with a known cost $c_i \geq 0$. As in our setting, a strategy $\pi$ at each round $t$ (starting from $t=1$) may choose to inspect a box $i \in [n]$ (or stay idle / inspect a dummy box) and proceed to round $t+1$, or halt. 
Crucially, not all uninspected boxes can be inspected at a given time $t$. In Constrained Pandora's Box, there is a predetermined collection of \emph{feasible sequences} of inspected boxes, denoted by $\mathcal{F}$, and we say that a strategy is \emph{feasible} if it adheres to $\mathcal{F}$. 
In other words, for any feasible strategy $\pi$, and for any random ordered tuple $(i_1, \dots, i_k)$ that can be generated by $\pi$ (i.e., is a subsequence of $S(\pi)$ for some $k \in [n]$), it holds that $(i_1, \dots, i_k) \in \mathcal{F}$. We say that $\mathcal{F}$ is \emph{prefix-closed} if for any tuple $(i_1, \dots, i_k) \in \mathcal{F}$ and any $j \in [k]$, it holds that $(i_1, \dots, i_j) \in \mathcal{F}$.
We denote an instance of this problem as $\InstanceConstraint = \left((c_i, V_i)_{i \in [n]}, \mathcal{F}\right)$; when $\mathcal{F}$ contains any possible sequence, i.e., when the problem is the unconstrained standard Pandora's Box, we just write $\InstanceConstraint = \left((c_i, V_i)_{i \in [n]}\right)$ instead. Let $\pi$ be a feasible strategy for $\InstanceConstraint$. For each $i \in [n]$, let $I_i(\pi)$ and $A_i(\pi)$ be the indicator random variables that signify whether box $i$ is inspected by $\pi$ and whether the reward of box $i$ is collected, respectively. Note that $A_i(\pi) \leq I_i(\pi)$ always holds because a box must be inspected by $\pi$ before its reward is collected. Since $\pi$ is a feasible strategy, the indicator random variables respect the exploration constraints imposed by $\mathcal{F}$. Finally, $A_i(\pi) = 1$ holds for the box with the maximum reward observed among those in $S(\pi)$. Therefore, we can write the random utility of this strategy $\pi$ as:
\begin{equation}\label{eq:utility-constrained}
u_{\InstanceConstraint}(\pi) := \sum_{i=1}^n A_i(\pi) V_i - \sum_{i=1}^n I_i(\pi) c_i = \max_{i \in S(\pi)} V_i - \sum_{i \in S(\pi)} c_i \,.
\end{equation}
\smallskip

\noindent\textbf{Characterization of Optimal Strategies via Reservation Values.} The work of \citet{KleinbergWW16} revived interest in the Pandora's Box problem and drew the attention of the economics and computation community towards it. Among other contributions, they provided a new proof of the optimality of Weitzman's rule for Pandora's Box. Later, \citet{Singla18} applied this proof to problems with constraints and different optimization objectives.

Let $\InstanceConstraint$ be an instance of Constrained Pandora's Box. We define the \emph{reservation value} $r_i$ of each box $i \in [n]$ to be the (unique) solution to the equation $\mathbb{E}_{V_i \sim D_i}[(V_i - r_i)^+] = c_i$ (where $(x)^+$ is a shortcut for $\max(x, 0)$). Moreover, for each box $i \in [n]$, let $Y_i := \min(V_i, r_i)$. 
We now state a result of \citet{KleinbergWW16} and provide its proof (which is essentially the original proof adjusted to our notation), in Appendix \ref{app:missing} for completeness. 
The statement refers to strategies for a class of stochastic optimization problems that generalizes the Pandora's Box problem in the sense that there still are $n$ boxes with costs and random rewards, $(c_i, V_i)_{i \in [n]}$, but the sequence of allowed inspections may be limited (e.g., by combinatorial constraints like in Constrained Pandora's Box), whereas the reward that may be collected is not necessarily restricted to the maximum observed value. In particular, it applies to all instances and strategies, such that the rewards are independent and only observed rewards can be collected.

\begin{lemma}[rf. Lemma 1 of \citet{KleinbergWW16}]
\label{lemma:kleinberg-lemma}
Let $(c_i, V_i)_{i \in [n]}$ and a strategy $\pi$ be such that: (1) the random variables  $(V_i)_{i \in [n]}$ are independent and (2) $A_i(\pi) \leq I_i(\pi)$ holds for every $i \in [n]$. Then,
\begin{equation}\label{eq:kleinberg-ineq}
     \EX\bigg[\sum_{i=1}^n A_i(\pi) V_i - \sum_{i=1}^n I_i(\pi) c_i\bigg]\leq \EX\bigg[\sum_{i=1}^n A_i(\pi)Y_i \bigg] .
\end{equation}
 Furthermore, inequality \eqref{eq:kleinberg-ineq} becomes an equality if $\pi$ has the following property: whenever it inspects a box $i$ and samples a value $V_i > r_i$, it accepts the box. 
\end{lemma}

Note that the lemma captures variants like Constrained Pandora's Box (the constraints can be simulated via $(I_i(\pi))_{i \in [n]}$) and Pandora's Box With Commitment (where $A_i(\pi)=1$ holds only for the last inspected box). 

Regarding the computational aspects of this work, our goal is to design strategies in polynomial time, assuming an appropriate value oracle. Given an instance $\Instance$, a vector of reservation values $(r_1,\dots, r_n)$, and a set of boxes $S \subseteq [n]$, the oracle outputs the expected maximum of the random variables $\min\{V_{i}, r_i\}$, for all $i \in S$. When the support of the distributions is polynomially bounded, this oracle can be simulated directly. In general, however, it can be estimated through sampling. For further details on how sample-based discrete distributions can be used to approximate continuous ones in our setting, we refer the reader to \citet{Guo0T021}.
\medskip

\noindent\textbf{Bounding the Adaptivity Gap.}
Consider an instance $\InstanceConstraint=\left((0, V_i)_{i \in [n]}, \mathcal{F} \right)$ of Constrained Pandora's Box with $\mathcal{F}$ being prefix-closed. Even though inspecting boxes incurs no cost in this stochastic optimization problem, finding the optimal adaptive strategy may be challenging due to the constraints imposed by $\mathcal{F}$. However, \citet{Singla18} showed that the expected utility of the optimal \emph{non-adaptive} strategy is a constant approximation of the optimal adaptive strategy, with the ratio being referred to as the \emph{adaptivity gap}. The constant was later improved by \citet{BradacSZ19}.

\begin{lemma}[due to \citet{BradacSZ19}]\label{lemma:adaptivity-gap}
Let $\InstanceConstraint=\left((0, V_i)_{i \in [n]}, \mathcal{F} \right)$ be an instance of Constrained Pandora's Box, where $\mathcal{F}$ is prefix-closed. Then,
\[
    \max_{S \in \mathcal{F}} \,\EX\big[\max_{i \in S} V_i \big] \geq \frac{1}{2} \max_{\pi \in \Pi} \,\EX\big[\max_{i \in S(\pi)} V_i \big]\,,
\]
where $\Pi$ is the set of all adaptive strategies for $\InstanceConstraint$.
\end{lemma}

\subsection{Prophet Inequality Problems}
A final class of stochastic optimization problems useful to our work is Prophet Inequality problems. In such problems, there is a set of random variables $Y_1, \dots, Y_n$, each drawn from a publicly known, non-negative distribution $D_i$. The distributions $D_1, \dots, D_n$ are independent, and in the standard Prophet Inequality setting, the order in which the random variables arrive is adversarial but known. That is, let $\sigma$ be a permutation of $(1, \dots, n)$ given to the decision maker (also known as ``gambler''). At round $i \in [n]$, the decision maker samples $Y_{\sigma(i)} \sim D_{\sigma(i)}$ and may either accept $Y_{\sigma(i)}$ or move on to the next random variable. Her goal is to approximate $\EX[\max_{j \in [n]} Y_j]$, i.e., to compete with a prophet who has access to all random realizations beforehand. Let $\InstanceProphet = (Y_i)_{i \in [n]}$ be an instance of the problem  and $\tau=(\tau_i)_{i \in [n]}$ be a collection of thresholds. We use $i^*_{\sigma}$ to denote the random index $\sigma(k) \in [n]$ for which $Y_{\sigma(k)} 
> \tau_{\sigma(k)}$ and $Y_{\sigma(j)} \leq \tau_{\sigma(j)}$ for all $j$ such that $j < k$. \citet{samuel84} showed a simple rule based on a single threshold that achieves a $2$-approximation.
\begin{lemma}[due to \citet{samuel84} and \citet{KleinbergW19}]\label{lemma:prophet-inequality}
     Let $\InstanceProphet=(Y_i)_{i \in [n]}$ be an instance of the Prophet Inequality problem, and let $\sigma$ be any permutation of the random variables. For thresholds $(\tau_i)_{i \in [n]}$ with $\tau_i=\frac{1}{2}\, \EX\big[\max_{j \in [n]}Y_j\big]$ for all $i \in [n]$, it holds that $\EX\left[Y_{i^*_{\sigma}}\right] \geq \frac{1}{2}\, \EX\left[\max_{j \in [n]}Y_j\right]$.
\end{lemma}
An important variant of the problem is the \textit{Free-Order Prophet Inequality} problem in which the decision maker is free to choose the permutation $\sigma$; note that this is the same as Pandora's Box With Commitment where all costs are zero. 
We summarize the state-of-the-art upper bound for the problem, due to \citet{BubnaC23}, in the following lemma.
\begin{lemma}[due to \citet{BubnaC23}]\label{lemma:fo-prophet-inequality}
    Let $\InstanceProphet=(Y_i)_{i \in [n]}$ be an instance of the Prophet Inequality problem. We can always find a permutation $\sigma$ of the random variables and thresholds $(\tau_i)_{i \in [n]} \in \mathbb{R}^n_{>0}$, such that
    $\EX\left[Y_{i^*_{\sigma}}\right] \geq 0.7258 \,\EX\Big[\max\limits_{j \in [n]}Y_j\Big]$.
\end{lemma}
We refer the reader to \citet{BubnaC23} for 
the exact closed-form expressions of the
thresholds $(\tau_i)_{i \in [n]}$ used 
in Lemma \ref{lemma:fo-prophet-inequality}.

\subsection{Monotone Submodular Functions}

We now define monotone submodular functions, as part of our problem (similarly as many related problems) will reduce to maximizing a monotone submodular function subject to combinatorial constraints.

\begin{definition}\label{def:mono-sm}
Given a ground set $[n]$, a function $f:2^{[n]} \to \mathbb{R}$ is \emph{monotone submodular} if: \emph{(i)} $f(S) \leq f(T)$, for every $S \subseteq T \subseteq [n]$ (monotonicity), and \emph{(ii)} $f(S)+f(T) \geq f(S \cup T) + f(S \cap T)$, for every $S, T \subseteq [n]$ (submodularity).
\end{definition}

Let $Y_1 ,\dots, Y_n$ be a collection of random variables, with each $Y_i$ drawn from a non-negative distribution $D_i$. Consider the set function defined as $f(S):= \EX\left[\max_{i \in S}Y_i\right]$ 
for each $S \subseteq [n]$ (using the convention that $\max_{i \in \emptyset}Y_i = 0$). 
It is not hard to show that $f$ is monotone submodular; see \Cref{app:missing} for a proof. Recall that $f$ is normalized if $f(\emptyset) = 0$. 
\begin{lemma}\label{lemma:f-monotone-submodular}
 The function  $f$  above is normalized, non-negative, and monotone submodular. 
\end{lemma}

\section[A Constant Approximation for Pandora's Box Over Time]{A Constant Approximation for Pandora's Box Over Time}\label{sec:general}

Next, we show how to derive a constant approximation to the guarantee of the optimal strategy for Pandora's Box Over Time in its full generality.
Note that the problem is NP-hard, as it is a generalization of the Free-Order Prophet Inequality problem, which is known to be NP-hard, even for distributions with support of size $3$ \citep{AgrawalSZ20}. The same is true even for the special cases studied in Section \ref{sec:special-cases}.

Theorem \ref{thm:general_result} is our main result, and the remainder of the section is dedicated to its proof.
\begin{theorem}
    \label{thm:general_result}
    There exists a strategy, $\MainStrategy$, for Pandora's Box Over Time which can be computed efficiently and provides a $21.3$-approximation to an optimal strategy. 
\end{theorem}
The proof is structured as follows.
In Section \ref{subsec:proxy}, for each Pandora's Box Over Time instance, $\Instance$, we identify a constrained instance, $\InstanceConstraint:=\InstanceConstraint(\Instance)$, with a carefully constructed prefix-closed constraint. We refer to $\InstanceConstraint$ as the \emph{proxy} instance for $\Instance$. This proxy instance will be central to devising a ``good'' strategy for our problem in Section \ref{subsec:main-strategy}, where we present our main result: a strategy that approximates the guarantee of the optimal strategy. To this end, we argue about how the performance of our strategy crucially depends on the performance of an algorithm for maximizing a monotone submodular function under a particular feasibility constraint. Finally, in Appendix \ref{subsec:crs}, we 
obtain a constant factor approximation algorithm for the submodular maximization problem in question. Combining these two results (Corollary \ref{cor:main-approx-star} and Theorem \ref{cor:block_matching}) yields Theorem \ref{thm:general_result}.

We often work with two instances (typically of different variants of the problem) and their respective strategies within the same proof. So, when needed for clarity, we write, e.g., $S(\Instance,\pi)$ and $S(\InstanceConstraint,\hat{\pi})$ instead of $S(\pi)$ and $S(\hat{\pi})$, respectively. 

\subsection{Reduction to Constrained Pandora's Box}\label{subsec:proxy}

Here we show that each Pandora's Box Over Time instance $\Instance$ has an equivalent representation as a Constrained Pandora's Box instance $\InstanceConstraint$ with certain exploration constraints. Moreover, we define a class of strategies for each such instance, each of which has a one-to-one correspondence with a strategy for the underlying Pandora's Box Over Time instance $\Instance$. The properties of $\InstanceConstraint$ and its associated strategy will be critical in our analysis. We begin by defining a hypergraph that we will later associate with  $\InstanceConstraint$.  

\begin{definition}\label{def:hypergraph-def}
    Given a Pandora's Box Over Time instance $\Instance$, its associated bipartite hypergraph $\mathcal{H}(\Instance) = (L, R, E)$ is defined as follows:
\begin{itemize}
    \item For each $i \in [n]$, there is a vertex $b_i \in L$.
    \item For each time slot $j \in [H]$, there is a vertex $t_j \in R$.
    \item For each $(i, j) \in [n] \times [H]$, there is a hyperedge $e(i, j) := \left\{ b_i \right\} \cup \{ t_k \mid k = j, \dots, \allowbreak j + p_i \}$ in $E$.
\end{itemize}
\end{definition}
Note that the construction of $\mathcal{H}(\mathcal{I})$ described above is done in polynomial time since the time horizon $H$ is part of the input (i.e., we are given $H$ distributions for each box $i \in [n]$.) To simplify the notation, when $\Instance$ is clear, we will sometimes write $\mathcal{H}$ instead of $\mathcal{H}(\Instance)$. We use $\mathcal{M}(\Instance)$ to denote the set of \emph{matchings} of $\mathcal{H}(\Instance)$, i.e., all the collections of disjoint hyperedges of $\mathcal{H}(\Instance)$. Note that while $\mathcal{H}(\Instance)$ is indeed a bipartite hypergraph (i.e., its vertices are partitioned into two sets, $L$ and $R$, such that $|L\cap e|\le 1$, for every hyperedge $e\in E$), it has even more structure. In particular,  for every hyperedge $e\in E$, the set $R\cap e$ consists of \emph{consecutive} (with respect to their index) vertices of $R$. We call such hypergraphs \emph{block bipartite}. This structure on the hyperedges of $\mathcal{H}(\Instance)$ will allow us to approximate the optimal solution to a submodular maximization problem on $\mathcal{M}(\Instance)$ in polynomial time; see the end of this section for the definition of the \emph{Submodular Block Matching} problem. 

\smallskip
\noindent\textbf{Proxy Instance and Proxy Strategy.} For each instance $\Instance$, we construct a \emph{proxy} Constrained Pandora's Box instance denoted by $\InstanceConstraint$. The set of boxes of $\InstanceConstraint$ is $[n] \times [H]$, i.e., for each box $i \in [n]$ and each time slot $t \in [H]$, we add a box labeled $(i, t)$ to $\InstanceConstraint$. Then, for each box $(i, t)$ we set the reward $V_{(i, t)}$ to follow the distribution $D_{it}$ and the cost of the box $(i, t)$ to be $c_{(i, t)} = \bar{c}_i(t)$. Finally, we construct the collection of ordered tuples $\mathcal{F}$ as follows: for each matching $\left\{e(i_1, j_1), \dots, e(i_k, j_k)\right\} \subseteq \mathcal{M}(\Instance)$, we add to $\mathcal{F}$ the ordered tuple of the boxes $\left( \left(i_{\sigma(1)}, j_{\sigma(1)}\right), \dots, \left(i_{\sigma(k)}, j_{\sigma(k)}\right) \right)$ where $\sigma$ is a permutation of the boxes so that $j_{\sigma(1)} < \dots < j_{\sigma(k)}$. Furthermore, for each strategy $\pi$ of $\Instance$, we define below a strategy $\hat{\pi}$ for $\InstanceConstraint$ which we call the \emph{proxy strategy} of $\pi$.
\begin{definition}\label{def:proxy-strategy}
    Let $\Instance$ be an instance of Pandora's Box Over Time and $\InstanceConstraint$ be its proxy Constrained Pandora's Box instance. Given a strategy $\pi$ for $\Instance$, we define its \emph{proxy strategy} $\hat{\pi}$ as the strategy that satisfies the following:
\begin{itemize}
    \item  $\hat{\pi}$ inspects box $(i, t) \in [n] \times [H]$ if and only if box $i \in S\left(\pi \right)$ and $t = t_i(\pi)$.
    \item  $\hat{\pi}$ stays idle at time $j \in [T_{\pi}]$ if and only if strategy $\pi$ stays idle at time $j$.
    \item  $\hat{\pi}$ halts at time $j \in [T_{\pi}]$ if and only if strategy $\pi$ halts at time $j$.
\end{itemize}
\end{definition}

The following lemma connects the two problems (the original and the proxy) and plays a key role in our analysis.
\begin{lemma}\label{lemma:reduction}
    Let $\Instance$ be an instance of Pandora's Box Over Time. For every strategy $\pi$ for $\Instance$, its proxy strategy $\hat{\pi}$ for $\InstanceConstraint$ is feasible. Furthermore, 
    \begin{equation}\label{eq:pihat-u-expr}
    \EX\Big[u_{\InstanceConstraint}(\hat{\pi})\Big] = \EX\bigg[ \max_{i \in S(\Instance,\pi)} V_{it_i(\pi)} - \sum_{i \in S(\Instance,\pi)} \bar{c}_i\left(t_i(\pi)\right)\bigg] \,.
    \end{equation}
\end{lemma}

\begin{proof}
    Let $(i_1, \dots, i_k)=S(\Instance, \pi)$. We first argue that strategy $\hat{\pi}$ is indeed feasible for $\InstanceConstraint$, i.e., it satisfies the sequence constraint $\mathcal{F}$ of $\InstanceConstraint$. By the above rules, the sequence of inspected boxes is
     ${S}\big(\InstanceConstraint, \hat{\pi}\big) = \left(\left(i_1, t_{i_1}(\pi)\right), \dots, \left(i_k, t_{i_k}(\pi)\right) \right) $.
    Clearly, the precedence requirement is satisfied, i.e., for any two boxes $j,j'$ with $i_j(\pi)< t_{i_{j'}}(\pi)$, box $(i_{j}, t_{i_j}(\pi))$ precedes box $(i_{j'}, t_{i_{j'}}(\pi))$ in ${S}\big(\InstanceConstraint, \hat{\pi}\big)$. In Claim \ref{claim:is-matching}, we show that the subgraph of $\mathcal{H}(\Instance)$ induced by $\big(\InstanceConstraint, \hat{\pi}\big)$ is a matching.
    \begin{claim}\label{claim:is-matching}
        The set  $M=\left\{e\left(i_1, t_{i_1}(\pi)\right), \dots, e\left(i_k, t_{i_k}(\pi)\right)\right\}$ is a matching.
    \end{claim}
    \begin{proof}
    Suppose, toward a contradiction, that $M$ is not a matching. This means that there is a pair $(i, t_i(\pi))$ and $(j, t_j(\pi))$ with $e(i, t_i(\pi)) \cap e(j, t_j(\pi)) \neq \emptyset$. Assume, without loss of generality, that $t_i(\pi) < t_j(\pi)$. Since each box in $\Instance$ can only be inspected once by $\pi$, by the first bullet of Definition \ref{def:proxy-strategy}, the same is true for $\InstanceConstraint$ and $\hat{\pi}$. Thus, clearly $i \neq j$, and therefore $b_i \neq b_j$. Furthermore, since strategy $\pi$ for $\Instance$ respects the processing time for both $i, j \in S(\pi)$, it holds that $t_j(\pi) \geq t_i(\pi) + p_i + 1$. If $p_i=0$, we trivially obtain that $e(i, t_i(\pi)) \cap e(j, t_j(\pi)) = \emptyset$, a contradiction. Consider now the case of $p_i \geq 1$. In this case, strategy $\pi$ stays idle in the interval $[t_i(\pi) + 1, t_i(\pi) + p_i]$.  However, by the second bullet of Definition \ref{def:proxy-strategy}, strategy $\hat{\pi}$ also stays idle in the interval $[t_i(\pi) + 1, t_i(\pi) + p_i]$. Consequently, by the construction of $\mathcal{H}(\Instance)$, there is no hyperedge $e' \in E$ such that $ e(i, t_i(\pi)) \cap e' \neq \emptyset$, which implies that $e(i, t_i(\pi)) \cap e(j, t_j(\pi)) = \emptyset$. This is a contradiction.
    \end{proof}    
    Now notice that, by the construction of $\big(\InstanceConstraint, \hat{\pi}\big)$, we have 
    \begin{align*}
        \EX\bigg[\max_{i \in S\left(\Instance, \pi \right)} V_{it_i(\pi)} - \!\!\!\sum_{i \in S\left(\Instance, \pi \right)} \!\!\bar{c}_i\left(t_i(\pi)\right) \bigg] &= \EX\bigg[\max_{\left(i, j\right) \in S\left(\InstanceConstraint, \hat{\pi}\right)} \!\! V_{\left(i, j\right)} - \!\!\!\!\sum_{\left(i, j\right) \in S\left(\InstanceConstraint, \hat{\pi}\right)} \!\!\! c_{\left(i, j\right)}\bigg] \\ &= \EX\left[ u_{\InstanceConstraint}\left(\hat{\pi}\right)\right],
    \end{align*}
    which concludes the proof.
\end{proof}
We continue by presenting an upper bound on the optimal expected utility of Pandora's Box Over Time instances, which will be useful in the analysis for our main result in Section \ref{subsec:main-strategy}. We slightly abuse notation and write $Y(e(i,j)):=Y_{(i,j)}=\min\left\{V_{(i,j)}, r_{(i,j)}\right\}$ for each hyperedge $e(i,j) \in E\left[\mathcal{H}\right]$ and box $(i,j) \in [n] \times [H]$.

\begin{lemma}\label{lemma:ub-sm-maximization}
    Let $\Instance$ be an instance of Pandora's Box Over Time and  $\pi^*$ be an optimal strategy for $\Instance$. Then,
         $\EX\left[u_{\Instance}(\pi^*)\right] \leq 2\cdot \!\!\max\limits_{M \in \mathcal{M}(\Instance)}  \!\EX\Big[\max\limits_{e(i,j) \in M}Y\left(e(i,j)\right)\Big]$.
\end{lemma}

\begin{proof}
    We have:
    \begin{align*}
        \EX[u_{\Instance}(\pi^*)] &= \EX\bigg[\max_{i \in S\left(\Instance, \pi^*\right)} \bar{v}_i\left(V_{it_i(\pi^*)}, T_{\pi^*} - t_i(\pi^*)-p_i\right) - \sum_{i \in S\left(\Instance, \pi^*\right)} \bar{c}_i\left(t_i(\pi^*)\right)\bigg] \\
        &\leq \EX\bigg[\max_{ i \in S\left(\Instance,  \pi^*\right)} V_{it_i(\pi^*)} - \sum_{i \in S\left(\Instance, \pi^*\right)} \bar{c}_i\left(t_i(\pi^*)\right)\bigg]  
        = \EX\left[u_{\InstanceConstraint}(\hat{\pi})\right] \\
        &\leq \EX\bigg[\sum_{(i,j) \in S\left(\InstanceConstraint,\hat{\pi}\right)} \!\!\!\!\! A_{(i,j)}(\hat{\pi}) Y_{(i,j)}\bigg] 
        \leq \EX\bigg[\sum_{(i,j) \in S\left(\InstanceConstraint, \hat{\pi}\right)} \!\!\!\! A_{(i,j)}(\hat{\pi})  \!\! \max_{(k,\ell) \in S\left(\InstanceConstraint, \hat{\pi}\right)} \!\!\!\!Y_{(k,\ell)}\bigg] \\ 
        &\leq \EX\bigg[\max_{(k,\ell) \in S\left(\InstanceConstraint, \hat{\pi}\right)} Y_{(k,\ell)}\bigg] 
        \leq \max_{\pi} \, \EX\bigg[\max_{(k,\ell) \in S\left(\InstanceConstraint, \pi\right)} Y_{(k,\ell)}\bigg] \numberthis\label{eq:pre-adaptivity-gap}.
    \end{align*}
    The first equality follows from the definition of the utility of a Pandora's Box Over Time instance in \eqref{eq:utility}. The first inequality follows from the fact that $\bar{v}_i$ as a function of its second argument only, attains a maximum at $0$, by definition. The second equality follows since, by Lemma \ref{lemma:reduction}, strategy $\hat{\pi}$ for the proxy constrained instance $\InstanceConstraint$ satisfies \eqref{eq:pihat-u-expr}. Then, we apply Lemma \ref{lemma:kleinberg-lemma} for $\big(\InstanceConstraint, \hat{\pi}\big)$ and obtain the second inequality due to \eqref{eq:kleinberg-ineq}. Finally, the fourth inequality follows since $\hat{\pi}$ inspects at most one box of $\InstanceConstraint$ for each random realization.

    Observe that the RHS of \eqref{eq:pre-adaptivity-gap} equals the optimal expected utility of the Constrained Pandora's Box instance with no costs $\big(\big(0, Y_{(i,j)}\big)_{(i,j) \in [n] \times [H]}, \mathcal{F}\big)$. Since $\mathcal{F}$ is a prefix-closed constraint, we can apply Lemma \ref{lemma:adaptivity-gap} for this instance and obtain:
    \begin{align}
         \max_{\pi} \, \EX\bigg[\max_{(k,\ell) \in S\left(\InstanceConstraint, \pi\right)} Y_{(k,\ell)}\bigg] &\leq 2\cdot \max_{S \in \mathcal{F}} \,  \EX\Big[\max_{(i,j) \in S} Y_{(i,j)}\Big]\nonumber\\ & = 2 \cdot \max_{M \in \mathcal{M}(\Instance)}  \EX\Big[\max_{e(i,j) \in M} Y\left(e(i,j)\right)\Big]\label{eq:adaptivity-gap-application} \,,
    \end{align}
    where the equality directly follows from the construction of the hypergraph $\mathcal{H}$.
    Combining \eqref{eq:pre-adaptivity-gap} and \eqref{eq:adaptivity-gap-application} completes the proof.
\end{proof}
Note that the last part of the proof, where Lemma \ref{lemma:adaptivity-gap} is invoked, is the reason why we need independence of reward distributions, not only across boxes, but also across rounds for each single box.

Given an instance $\Instance$ and its associated hypergraph $\mathcal{H} = \mathcal{H}(\Instance)$, consider the set function $f: 2^{E\left[\mathcal{H}\right]} \to \mathbb{R}_{\geq 0}$ defined as 
\begin{equation}\label{eq:submod_obj}
    f(M):=\EX\Big[\max_{e(i,j) \in M} Y_{(i,j)} \Big], \text{\ \ for each \ \ } M \subseteq  E\left[\mathcal{H}\right] \,.
\end{equation}
By Lemma \ref{lemma:f-monotone-submodular}, $f$ is a non-negative monotone submodular function with ground set $E\left[\mathcal{H}\right]$. Under this perspective, we can observe that
\begin{equation*}
 \max_{M \in \mathcal{M}(\Instance)} \EX\Big[\max_{e(i,j) \in M} Y\left(e(i,j)\right)\Big] = \max_{M \in \mathcal{M}(\Instance)} f(M) 
\end{equation*}
and interpret the inequality of Lemma \ref{lemma:ub-sm-maximization} as follows: the optimal expected utility of a Pandora's Box Over Time instance is upper-bounded by $2$ times the optimal solution of an instance of a monotone submodular maximization problem subject to a matching in a block bipartite hypergraph constraint. We conclude the section with the formal statement of this optimization problem.

\smallskip
\noindent\textbf{The Submodular Maximization Problem.} Let $\mathcal{H} = (L, R, E)$ be a block bipartite hypergraph,  $\mathcal{M} \subseteq 2^{E}$ be its set of matchings, and  $f: 2^{E} \to \mathbb{R}_{\geq 0}$ be a normalized monotone submodular function. The {Submodular Matching on Block Bipartite Hypergraphs} problem, or \emph{Submodular Block Matching} for short, asks for a matching $M^* \in \argmax_{M \in \mathcal{M}} f(M)$.
In Appendix \ref{subsec:crs}, we show how to get a polynomial-time $5.32$-approximation algorithm for the problem. 

\begin{restatable}{theorem}{block}
\label{cor:block_matching}
There is a polynomial-time $5.32$-approximation algorithm for Submodular Block Matching. 
\end{restatable}

\subsection{Our Order-Non-Adaptive Strategy}\label{subsec:main-strategy}
In this section, we present $\MainStrategy$ (Strategy \ref{strategy:main}), a simple strategy for Pandora's Box Over Time, which we show to be a {$21.3$}-approximation of the optimal expected utility.

\begin{strategy}[ht]
\DontPrintSemicolon
\caption{\MainStrategy\label{strategy:main}}
{
\nonl $\hspace{-2.3ex}\rhd$ {\bf{Input:}} An instance $\Instance = \left(\bar{c}_i, p_i, (V_{it})_{t \in [H]}, \bar{v}_i\right)_{i \in [n]}$ and an $\alpha$-approximation algorithm $\textsc{Alg}$ for Submodular Block Matching. \;
\tcp{Phase 1: determine schedule of inspection and threshold}
Construct the bipartite hypergraph $\mathcal{H}(\Instance)=(L, R, E)$ as in Definition~\ref{def:hypergraph-def}.\\
Let $f(\cdot)$ be the monotone submodular function on $2^E$ defined in \eqref{eq:submod_obj}.\\
Find matching $M=\left\{e(i_1, t_1), \dots, e(i_k, t_k)\right\} \subseteq E$ using $\textsc{Alg}$ on instance with objective $f$ and constraints induced by the matchings of $\mathcal{H}(\Instance)$. \hspace{-3pt}\tcp*{{\scriptsize  $t_{1} < \dots < t_{k}$}}
Set $\tau=0.5 \cdot \EX\left[\max_{e(i,t)\in M} Y(e(i,t)) \right]$.\\
\tcp{{Phase 2: threshold-based strategy using the schedule and threshold of Phase 1}}
\For{$\ell=1,\dots, k$}{
    \If{$r_{\left(i_{\ell}, t_{\ell}\right)} > \tau$}{\label{strategy:main:opened}
        Sample $V_{\left(i_{\ell}, t_{\ell}\right)} \sim D_{i_{\ell} t_{\ell}}$ at time $t_{\ell}$.\\
        \If{$V_{\left(i_{\ell}, t_{\ell}\right)} > \tau$}{\label{strategy:main:accepted}
            \textbf{halt}\label{strategy:main:halt} \tcp*{{the reward of the last inspected box is collected}}
        }
    }
}
}
\end{strategy}

Given an instance of Pandora's Box Over Time, the first step of $\MainStrategy$ (Phase $1$) is to determine a preliminary \emph{schedule} of inspection times \emph{before} inspecting a single box. Note that the strategy may eventually not inspect all the boxes in the preliminary set, as it may halt sooner based on a threshold-based stopping rule we specify in Phase 2. Such strategies are called order-non-adaptive in the literature \citep[see, e.g.,][]{BeyhaghiCai23}.
The stopping rule we specify is inspired by the approach of \citet{EsfandiariHLM19}. The crux of this approach is to relate the expected utility of our Pandora's Box Over Time instance to the expected utility achieved by the ``gambler'' in an instance $\InstanceProphet$. Once this is accomplished, we then relate the performance of our algorithm to the performance of the prophet.
\begin{lemma}\label{lemma:u-approx-sm}
    Let $\Instance$ be an instance of Pandora's Box Over Time and let $\textsc{Alg}$ be an $\alpha$-approximation algorithm for Submodular Block Matching. It holds that
    \begin{equation*}
        \EX[u_{\Instance}(\MainStrategy)] \geq \frac{1}{2\alpha} \cdot \max_{M \in \mathcal{M}(\Instance)}  \EX\Big[\max_{e(i,j) \in M}Y\left(e(i,j)\right)\Big]\,.
    \end{equation*}
\end{lemma}
\begin{proof}
    Denote $\pi := \MainStrategy$ for brevity. 
    Moreover,  let $\hat{\pi}$ be the proxy strategy of $\pi$ for $\InstanceConstraint$. We have:
    \begin{align*}
         \EX[u_{\Instance}(\pi)] &= \EX\bigg[\max_{i \in S\left(\Instance, \pi\right)} \bar{v}_i\left(V_{it_i(\pi)}, T_{\pi} - t_i(\pi)-p_i\right) - \sum_{i \in S\left(\Instance, \pi\right)} \bar{c}_i\left(t_i(\pi)\right)\bigg] \\
         &= \EX\bigg[\max_{i \in S\left(\Instance, \pi\right)} \bar{v}_i\left(V_{it_i(\pi)}, 0\right) - \sum_{i \in S\left(\Instance, \pi\right)} \bar{c}_i\left(t_i(\pi)\right)\bigg]\\
         &= \EX\bigg[\max_{i \in S\left(\Instance, \pi\right)} V_{it_i(\pi)} - \sum_{i \in S\left(\Instance, \pi\right)} \bar{c}_i\left(t_i(\pi)\right)\bigg]=\EX\left[u_{\InstanceConstraint}(\hat{\pi})\right]\,.\label{eq:before-kleinberg-general}\numberthis
    \end{align*}
    The first equality follows from the definition of the expected utility of a Pandora's Box Over Time instance in \eqref{eq:utility}. The second equality follows from the fact that $\pi$ ``collects'' at round $T_{\pi}$ the reward of the box inspected at time $t_i(\pi)=T_{\pi}-p_i$; indeed, $\pi$ halts at  $T_{\pi}$ and at that time the last reward is the only one exceeding $\tau$, so it is the one collected. The third equality follows from the definition of the function $\bar{v}_i$ for $i \in [n]$. Finally, the fourth equality holds since, by Lemma \ref{lemma:reduction}, strategy $\hat{\pi}$ for $\InstanceConstraint$ satisfies \eqref{eq:pihat-u-expr}. To continue, we show that $(\InstanceConstraint, \hat{\pi})$ satisfies the second condition of Lemma \ref{lemma:kleinberg-lemma}; in particular, we show that here \eqref{eq:kleinberg-ineq} holds with equality.
    
\begin{claim}\label{claim:kleinberg-equality-general}
    It holds that
       $ \EX\left[u_{\InstanceConstraint}(\hat{\pi})\right] = \EX\big[\sum_{(i,t) \in [n] \times [H]} A_{(i,t)}(\hat{\pi}) Y_{(i,t)} \big]$.
\end{claim}

\begin{proof}
    According to Lemma \ref{lemma:kleinberg-lemma}, this identity holds if the following condition is satisfied for 
   $\big(\InstanceConstraint, \hat{\pi} \big)$:
     \[\left( \exists (i,t) \in [n] \times [H] : I_{(i,t)}(\hat{\pi}) = 1 \,\land\, V_{(i,t)} > r_{(i,t)} \right) \Rightarrow A_{(i,t)}(\hat{\pi}) = 1\,.
    \]
    By the first bullet of Definition \ref{def:proxy-strategy}, $I_{(i,t)}(\hat{\pi}) = 1$ if and only if $i \in S(\pi)$ and $t_i(\pi) = t$. This happens only if the condition on Line \ref{strategy:main:opened} of $\MainStrategy$ for this pair of $(i,t)$ is evaluated to \textbf{true}. Therefore, $r_{(i,t)} > \tau$. Since, by assumption, we additionally have that $V_{(i,t)} > r_{(i,t)}$, we can conclude that $V_{(i,t)} > \tau$. However, this implies that the condition on Line \ref{strategy:main:accepted} of $\MainStrategy$ is evaluated to \textbf{true} and therefore, strategy $\pi$ halts (by Line \ref{strategy:main:halt}). By the third bullet of Definition \ref{def:proxy-strategy}, strategy $\hat{\pi}$ halts as well. Note that, for each box $j \in S(\pi)$ with $t_j(\pi)< t$ it holds that $V_{(j, t_j(\pi))} \leq \tau < V_{(i, t)}$.  
    Therefore, $A_{(i,t)}(\hat{\pi})=1$ (the reward of box $(i,t)$ in $\InstanceConstraint$ is the maximum among inspected boxes), and the claim follows.
    \end{proof}
    
    Let $M$ be the matching returned by $\textsc{Alg}$ with input
    the function $f$ (as in \eqref{eq:submod_obj}) and the hypergraph
    $\mathcal{H}(\Instance)$, and let $i^* =: i^*(\pi)$ be the random variable denoting the box inspected at time $t^*=T_{\pi}-p_{i^*}$. Clearly, by Line \ref{strategy:main:opened} and by Line \ref{strategy:main:accepted}, it holds that $Y_{(i^*, t^*)}= \min \left\{r_{(i^*, t^*)}, V_{(i^*, t^*)}\right\} > \tau$ 
    and $Y_{(i_j, t_j)}\leq V_{(i_j, t_j)}\leq \tau$ for all $j$ such that $e(i_j, t_j) \in M$ and $i_j \leq i^*$.
    
    Furthermore, $i^*(\pi)$ is the only box in $S(\pi)$ for which this holds. By invoking Claim \ref{claim:kleinberg-equality-general}, we can continue \eqref{eq:before-kleinberg-general} as follows:
    \begin{align*}
         \EX[u_{\Instance}(\pi)] &= \EX\bigg[\sum_{(i,t) \in [n] \times [H]} A_{(i,t)}(\hat{\pi}) Y_{(i,t)} \bigg]=\EX\left[Y_{(i^*,t^*)}\right]=\EX\left[Y(e(i^*,t^*))\right]\\
         &\geq\frac{1}{2} \cdot \EX\Big[\max_{e(i,j) \in M}Y_{e(i,j)}\Big] \geq \frac{1}{2\alpha} \cdot \max_{M \in \mathcal{M}(\Instance)}  \EX\Big[\max_{e(i,j) \in M}Y\left(e(i,j)\right)\Big]\,.
    \end{align*}
    The first inequality follows by observing that $\EX\left[Y(e(i^*,t^*))\right]$ equals the expected value of the gambler for the instance $\InstanceProphet=(Y_{e(i,j)})_{e(i,j) \in M}$; therefore, the inequality holds due to Lemma~\ref{lemma:prophet-inequality}. Finally, the second inequality follows from the fact that the matching $M$ is the solution of an $\alpha$-approximation algorithm for this precise objective. The proof follows.
\end{proof}

Combining Lemmata \ref{lemma:ub-sm-maximization} and \ref{lemma:u-approx-sm}, we directly get the following general result that relates the guarantee of our order-non-adaptive strategy with the quality of approximation we can achieve for Submodular Block Matching.

\begin{corollary}\label{cor:main-approx-star}
    Let $\Instance$ be an instance of Pandora's Box Over Time, $\pi^*$ be an optimal strategy for $\Instance$, and  $\textsc{Alg}$ be an $\alpha$-approximation algorithm for Submodular Block Matching. It holds that
    \(
        \EX[u_{\Instance}(\MainStrategy)] \geq \frac{1}{4\alpha} \cdot \EX\left[u_{\Instance}(\pi^*)\right]\,.
    \)
\end{corollary}

Of course, Corollary \ref{cor:main-approx-star} is a conditional version of our Theorem \ref{thm:general_result}. The proof of the latter follows, by using in Corollary \ref{cor:main-approx-star}  the $5.32$-approximation algorithm for Submodular Block Matching of Theorem \ref{cor:block_matching}.

\section{Three Natural Special Cases}\label{sec:special-cases}

Our Pandora's Box Over Time problem is, by design, very general. This is clear from the discussion on all the different variants of related problems it captures as special cases. It is reasonable to expect that there are several other meaningful restricted versions of Pandora's Box Over Time worth studying. Here we turn to three such special cases; one where the processing times are all zero, one in which boxes are available only at certain time slots and the costs do not change over time,\footnote{As we mention in the Introduction and we discuss in detail in Section \ref{subsec:timeslots} below, technically this is not a special case of Pandora's Box Over Time. Yet, it does reduce to a special case Pandora's Box Over Time in a natural way when extreme costs are used to simulate the box unavailability constraints.} and another one where neither the cost nor the distribution of each box changes over time. All restrictions align with the majority of the related literature, and in all three cases, we can significantly improve over Theorem \ref{thm:general_result}.

\subsection{Pandora's Box Over Time With Instant Inspection}\label{subsec:instant}
Although the use of processing times was suggested by \citet{Weitzman79} 
in the paper that introduced the Pandora's Box problem, their effect has not been studied before. It is true that the varying processing times create technical complications, forcing us to work with matchings in non-uniform hypergraphs. Removing the processing times is a very natural restriction of Pandora's Box Over Time, resulting in what we call \emph{Pandora's Box Over Time With Instant Inspection}. The general instance of this problem is $\Instance = \left(\bar{c}_i, 0, (V_{it})_{t \in [H]}, \bar{v}_i\right)_{i \in [n]}$.

\citet{AmanatidisFRT24} showed that it is possible to efficiently compute a $(8 + \varepsilon)$-approximation for this special case. The main observation there was that $\mathcal{H}(\Instance)$  is not a hypergraph anymore, but a bipartite graph instead. In this case, Submodular Block Matching becomes the much better understood Submodular Matching on Bipartite Graphs problem. So, instead of Theorem \ref{cor:block_matching}, one could rather invoke the $(2+\delta)$-approximation algorithm of \citet{LeeSV10} for maximizing a submodular function subject to the intersection of two matroids (bipartite matching being a special case of that). However, specifically for maximizing the objective $f(\cdot)$ of \eqref{eq:submod_obj} subject to a bipartite matching constraint, here we show how to obtain a $(4+\varepsilon)$-approximation. To do so, in place of invoking the algorithm of \citet{LeeSV10}, we use the next lemma.

\begin{restatable}{lemma}{stochastic}\label{lemma:exp-max-matching}
For any $\delta > 0$, there is a polynomial time $(1+\delta)$-approximation algorithm for the stochastic problem of maximizing $f(\cdot)$ of \eqref{eq:submod_obj} when $\mathcal{H}(\Instance)$ is a graph.
\end{restatable}

Lemma \ref{lemma:exp-max-matching} follows from a general framework for obtaining randomized polynomial time approximation schemes due to \citet{ChenHLLLL16} and \citet{LiD19}. For details, see Appendix \ref{app:PTAS}.

\begin{theorem}
\label{theorem:instant-proc-main}
Fix any constant $\varepsilon>0$. There exists a strategy for Pandora's Box Over Time With Instant Inspection (a variant of $\MainStrategy$ (Strategy \ref{strategy:main})) which can be computed efficiently and provides a $(4+\varepsilon)$-approximation to an optimal strategy.
\end{theorem}
\begin{proof}
Let $\Instance$ be an instance of Pandora's Box Over Time With Instant Inspection. In this case, $\MainStrategy$ in Phase 1 would construct the \textit{graph} $\mathcal{H}(\Instance)$ and run the $(1+\delta)$-approximation algorithm of Lemma \ref{lemma:exp-max-matching} for $\delta = \nicefrac{\varepsilon}4$, to obtain the matching $M$. From that point on, everything is analogous to before.
Combining Corollary \ref{cor:main-approx-star} with this $(1+ \nicefrac{\varepsilon}4)$ factor, completes the proof.
\end{proof}

\subsection{The Pandora's Box Problem With Time Slots}\label{subsec:timeslots}

Even though instances of Pandora's Box Over Time are not coupled with constraints on the sequences of inspected boxes, our framework is rich enough to capture a range of problems with feasibility constraints. In this section, we highlight this by focusing on the class of instances in which boxes are available for inspection only at certain rounds. This setting is a generalization of the one studied by \citet{BergerEFF24}.

Formally, there are $n$ boxes and $H$ rounds. Each box $i \in [n]$, in addition to having reward distributions $D_{it}$ for $t = 1, \dots, H$, an inspection cost $c_i$, and a discounting function $\bar{v}_i$,\footnote{In the original definition of the Pandora's Box With Time Slots problem by \citet{BergerEFF24} there was also the notion of a deadline per box $i$ (which is already captured by $\Phi_i$). Furthermore, there was no discounting and, therefore, we slightly generalize the definition here.} is also associated with a set of time slots $\Phi_i \subseteq [H]$, which are the only rounds for which box $i$ is available for inspection. We call this class of problems \emph{Pandora's Box With Time Slots}. An instance of this problem is characterized by
$\InstanceInterval = \big(c_i, \left(V_{it}\right)_{t \in [H]}, \Phi_i, \bar{v}_i\big)_{i \in [n]},
$ and a strategy $\pi_{\textsc{sl}}$ is feasible for $\InstanceInterval$ if, for any random tuple of inspected boxes $(i_1, \dots, i_k)$ generated by $\pi_{\textsc{sl}}$, the condition
$t_{i_j}(\pi) \in \Phi_{i_j}$ for $j = 1, \dots, k$ is satisfied.

It is not hard to see that instance $\InstanceInterval$ can be implicitly captured by our framework. In particular, for every instance $\InstanceInterval$, we can construct an auxiliary instance $\Instance=\big(\bar{c}_i, 0, (V_{it})_{t \in [H]}, \bar{v}_i\big)_{i \in [n]}$ 
 with
\begin{equation*}
    \bar{c}_i(t) =
\begin{cases}
c_i & \text{if } t \in\Phi_i, \\
M & \text{otherwise,}
\end{cases}
\end{equation*}
where $M$ is a prohibitively high inspection cost for any box (e.g., $M > \displaystyle\max_{i \in [n]}\displaystyle\max_{t \in [H]}\EX[V_{it}]$).

Lemma \ref{lemma:time-slots-equiv} encodes the easy equivalence between an instance $\InstanceInterval$ and its auxiliary instance $\Instance$ and will be important for the proof of Theorem \ref{theorem:time-slots-main} that follows.
\begin{lemma}\label{lemma:time-slots-equiv}
   Let $\InstanceInterval$ be an instance of Pandora's Box With Time Slots and $\Instance$ be the auxiliary Pandora's Box Over Time instance.
   The following statements are true.
   \begin{enumerate}[i.]
       \item Every strategy $\pi$ for $\Instance$ that never inspects a box $i$ in a round $t \not\in \Phi_i$ is \emph{equivalent} to a feasible strategy $\pi_{\textsc{sl}}$ of $\InstanceInterval$ and vice versa.
       \item For every strategy $\pi$ for $\Instance$, we can construct a strategy $\pi'$ for $\Instance$ that never inspects a box $i$ in a round $t \not\in \Phi_i$ and \emph{weakly dominates} $\pi$ for every sequence of boxes $S(\pi)$ generated by $\pi$ i.e., $\EX[u_{\Instance}(\pi')] \geq \EX[u_{\Instance}(\pi)].$
   \end{enumerate}
\end{lemma}
\begin{theorem}
\label{theorem:time-slots-main}
Fix any constant $\varepsilon>0$. There exists a strategy for Pandora's Box With Time Slots (a variant of $\MainStrategy$ (Strategy \ref{strategy:main})) which can be computed efficiently and provides a $(4+\varepsilon)$-approximation to an optimal strategy.
\end{theorem}

\begin{proof}
Fix an instance $\InstanceInterval$, and let $\pi_{\textsc{sl}}^*$ be an optimal strategy for $\InstanceInterval$. Let $\Instance$ denote the auxiliary Pandora's Box Over Time instance defined above. 
Since $\Instance$ is such that $p_i = 0$ for each box $i \in [n]$, we can approximate an optimal strategy $\pi^*$ for it using the variant of $\MainStrategy$ (Strategy \ref{strategy:main}) described in Section \ref{subsec:instant}. Additionally, consider a strategy $\pi'_{\textsc{main}}$, which weakly dominates $\MainStrategy$, as described in Lemma \ref{lemma:time-slots-equiv} (statement ii.). We have:
\begin{align*}
    \EX[u_{\InstanceInterval}(\pi_{\textsc{sl}}^*)] \leq \EX[u_{\Instance}(\pi^*)] \leq (4+\varepsilon) \EX[u_{\Instance}(\pi_{\textsc{main}})] \leq (4+\varepsilon) \EX[u_{\Instance}(\pi'_{\textsc{main}})].
\end{align*}
The first inequality follows from Lemma \ref{lemma:time-slots-equiv} (statement i.), which states that the optimal strategy $\pi_{\textsc{sl}}^*$ for $\InstanceInterval$, being feasible, is equivalent to a strategy for $\Instance$, the expected utility of which is upper-bounded by that of $\pi^*$. The second inequality is due to Theorem \ref{theorem:instant-proc-main}, and the final inequality follows from Lemma \ref{lemma:time-slots-equiv} (statement ii.). This completes the proof of the theorem since Lemma \ref{lemma:time-slots-equiv} (statement i.) implies that the strategy $\pi'_{\textsc{main}}$, which by construction never inspects a box $i$ in a round $t \not \in \Phi_i$, is equivalent to a feasible strategy for $\InstanceInterval$ with an identical expected utility for the decision maker.
\end{proof}

\subsection{Pandora's Box Over Time With Constant Costs and Distributions}

The aspects of Pandora's Box Over Time that give it the most flexibility are probably the time-dependent costs and distributions. Just by varying the cost functions accordingly, it is easy to simulate a large number of scenarios, like deadlines, time windows, and knapsack constraints with respect to time. It is thus reasonable to consider the simpler class of instances where boxes still have processing times and the values of inspected boxes may degrade over time, but both the costs and the distributions of the rewards are not time-dependent. We call this restriction \emph{Pandora's Box With Value Discounting}.
Formally, for each $i \in [n]$ and each $t \in[H]$, we have that $D_{it}=D_i$. Furthermore, for each $i \in [n]$ and each $t \in \mathbb{Z}_{>0}$, we have $\bar{c}_i(t) = c_i$. That is, we consider instances of the form $\Instance = \left(c_i, p_i, V_i , \bar{v}_{i}\right)_{i\in [n]}$.

\begin{theorem}\label{theorem:discounting-main}
There exists a strategy for Pandora's Box With Value Discounting (see $\FixedCostStrategy$ (Strategy \ref{strategy:fixed-costs})) which can be computed efficiently and provides a $1.37$-approximation to an optimal strategy.
\end{theorem}

 We begin by upper-bounding the optimal expected utility. To this end, recall the definition of the random variables $Y_i = \min\{V_i,r_i\}$, where $r_i$ is the reservation value of the $i$-{th} box.

\begin{lemma}\label{lemma:opt-prophet-ub}
    Let $\Instance$ be an instance of Pandora's Box With Value Discounting. Then, $\EX[u_{\Instance}(\pi^*)] \leq \EX[\max_{i \in [n]} Y_i]$.
\end{lemma}
\begin{proof}
    Consider the instance $\InstanceConstraint = \left((c_i, V_i)_{i \in [n]}\right)$ of the Pandora's Box problem. We have:
    \begin{align*}
        \EX[u_{\Instance}(\pi^*)] &= \EX\bigg[\max_{i \in S(\pi^*)} \bar{v}_i\left(V_i, T_{\pi^*} - t_{i}(\pi^*)-p_i \right) - \!\sum_{i \in S(\pi^*)} \!c_i\bigg] \\
        &\leq \EX\bigg[\max_{i \in S(\pi^*)} V_i - \!\sum_{i \in S(\pi^*)} \!c_i \bigg]
        = \EX\left[u_{\hat{\Instance}}(\pi^*)\right]  \\
        &\leq \EX\bigg[\sum_{i=1}^n A_i(\pi^*) Y_i \bigg]\leq \EX\Big[\max_{i \in [n]} Y_i \Big]\,.
    \end{align*}
    The first equality follows from \eqref{eq:utility} for $(\Instance, \pi^*)$ (simplified for this special class of instances). The first inequality follows from the fact that $\bar{v}_i$ as a function of its second argument only, attains a maximum at $0$. The subsequent equality follows from \eqref{eq:utility-constrained}, noting that $\pi^*$ is a feasible strategy for $\InstanceConstraint$. Then, we apply Lemma \ref{lemma:kleinberg-lemma} for $(\InstanceConstraint, \pi^*)$, and therefore the second-to-last inequality holds due to \eqref{eq:kleinberg-ineq}. The lemma follows since $\pi^*$ may ``collect'' the reward of at most one box for each observed $\vec{V} \sim \vec{D}$, i.e., $\sum_{i=1}^n A_i(\pi^*) \leq 1$ holds by the definition of the problem. 
\end{proof}

\noindent
    An interpretation of Lemma \ref{lemma:opt-prophet-ub} is that the optimal expected utility of any instance $\Instance$ is upper bounded by the optimal value achieved by the prophet for the instance $\InstanceProphet = (Y_i)_{i \in [n]}$. Working similarly to Section \ref{sec:general}, we design an order-non-adaptive threshold-based strategy for the problem. However, unlike the general case, here the ``proxy'' instance related to our Pandora's Box Over Time instance is simply an instance of (the unconstrained) Pandora's Box problem. This constraint-free environment allows us to rely on a Free-Order Prophet Inequality instance to generate an order of the boxes, and this drastically improves performance.
\begin{strategy}[ht]
\DontPrintSemicolon
\caption{\FixedCostStrategy\label{strategy:fixed-costs}}
{
\nonl $\hspace{-2.3ex}\rhd$ {\bf{Input:}} An instance $\Instance = \left(c_i, p_i, V_i, \bar{v}_i\right)_{i \in [n]}$ of Pandora's Box} With Value Discounting \;
\tcp{Phase 1: determine schedule of inspection and thresholds}
Construct $\InstanceProphet=(Y_i)_{i \in [n]}$. \label{line:Yi}\\
Obtain a permutation $\sigma$ of $[n]$ and thresholds $(\tau_{i})_{i \in [n]} \in \mathbb{R}^n_{>0}$ that achieve the guarantee of Lemma \ref{lemma:fo-prophet-inequality} 
for $\InstanceProphet$. \label{line:sigma}\\
\tcp{Phase 2: threshold-based strategy using the schedule and thresholds of Phase 1}
Set $t=1$\\
\For{$i=1,\dots, n$}{
    \If{$r_{\sigma(i)} > \tau_{\sigma(i)}$}{\label{strategy:fixed:opened}
        Sample $V_{\sigma(i)} \sim D_{\sigma(i)}$ at time $t$.\\
        \If{$V_{\sigma(i)} > \tau_{\sigma(i)}$} {\label{strategy:fixed:accepted}
            \textbf{halt}\ \tcp*{the reward of the last inspected box is collected}
        }
        Set $t=t+1+p_{\sigma(i)}.$  \tcp*{wait for $p_{\sigma(i)}$ rounds}
    }
}
\end{strategy}

We next show that the expected utility of $\FixedCostStrategy$ for an instance $\Instance = (c_i, p_i, V_i, \allowbreak \bar{v}_i)_{i \in [n]}$ is at least the expected value achieved by the ``gambler'' for the related instance $\InstanceProphet = (Y_i)_{i \in [n]}$.

\begin{lemma}\label{lemma:u-is-fo}
    Let $\Instance$ be an instance of the Pandora's Box With Value Discounting,  and $\InstanceProphet$, $\sigma$ be as in  \Cref{line:Yi,line:sigma} of $\FixedCostStrategy$. Then, $\EX[u_{\Instance}(\FixedCostStrategy)] \geq \EX[Y_{i^*_{\sigma}}]$.
\end{lemma}
\begin{proof}
    We use $\pi := \FixedCostStrategy$ for brevity. Let ${v_i}': \mathbb{R}_{\geq 0} \times \mathbb{N} \mapsto \mathbb{R}_{\geq 0}$ be such that $v_i'(x, 0)= x$ for all $x \geq 0$ and $v'_i(\cdot, t)=0$ for all $t>0$. Further let $\Instance'=(c_i, p_i, V_i, \allowbreak v'_i)_{i \in [n]}$.\footnote{It is straightforward to see that $\Instance'$ is equivalent to an instance of Pandora's Box With Commitment, i.e., the version of Pandora's Box in which only the last box is collected (see, e.g., \citet{FuLX18}). This observation is crucial, as we later invoke Lemma \ref{lemma:kleinberg-lemma} for $(\Instance', \pi)$.} By the definition of $v_i'(\cdot)$, $\bar{v}_i(x,t) \geq v'_i(x,t)$ holds for every $x \geq 0$ and every t $\in \mathbb{N}$. Therefore,
    \begin{align*}
        \EX[u_{\Instance}(\pi)] &= \EX\bigg[\max_{i \in S(\pi)} \bar{v}_i\left(V_i, T_{\pi} - t_i(\pi)-p_i \right) - \sum_{i \in S(\pi)} c_i \bigg] \\ 
        &\geq \EX\bigg[\max_{i \in S(\pi)} v_i'\left(V_i, T_{\pi} - t_i(\pi) -p_i\right) - \sum_{i \in S(\pi)} c_i \bigg]  = \EX[u_{\Instance'}(\pi)]\,.\numberthis \label{eq:utility-equal-fixed}
    \end{align*} 
    For each $i \in [n]$, let $A_i(\pi)$ and $I_i(\pi)$ be the indicator random variables that signify whether the reward of box $i$ is collected and whether the box is opened in $\Instance'$. In Claim \ref{claim:kleinberg-equality-fixed}, we show that $(\Instance', \pi)$ is such that \eqref{eq:kleinberg-ineq} holds with equality.

    \begin{claim}\label{claim:kleinberg-equality-fixed}
        It holds that
           $\EX[u_{\Instance'}(\pi)] = \EX\big[\sum_{i=1}^n A_i(\pi) Y_i \big]$.
    \end{claim}

    \begin{proof}
    According to Lemma \ref{lemma:kleinberg-lemma}, this identity holds if the following condition is satisfied for $\big(\Instance', \pi \big)$:
    $\left( \exists i \in [n] : I_i(\pi) = 1 \land V_i > r_i \right) \Rightarrow A_i(\pi) = 1$.
    Indeed, by Line \ref{strategy:fixed:opened} of $\pi$, $I_{\sigma(i)}(\pi) = 1$ holds for a box $\sigma(i) \in [n]$ if $r_{\sigma(i)} > \tau_{\sigma(i)}$. If we additionally have that $V_{\sigma(i)} > r_{\sigma(i)}$, which holds if the condition on Line \ref{strategy:fixed:accepted} is evaluated as \textbf{true}, the strategy halts. By the construction of $\Instance'$ (and, in particular, of $v'(\cdot)$), only the box $\sigma(i)$ is such that $V_{\sigma(i)} > 0$. Hence, $A_{\sigma(i)}(\pi) = 1$, and the claim follows.      
    \end{proof}
    By invoking Claim \ref{claim:kleinberg-equality-fixed}, we can continue \eqref{eq:utility-equal-fixed} as follows:
    \begin{equation*}
         \EX[u_{\Instance'}(\pi)] = \EX\bigg[\sum_{i=1}^n A_i(\pi) Y_i \bigg] = \EX[Y_{i^*_\sigma}]\,.
    \end{equation*}
    The last equality follows from the fact that $A_i(\pi) = 1$ holds for the \emph{first} box $\sigma(i)$ for which $r_{\sigma(i)} > \tau_{\sigma(i)}$ (Line \ref{strategy:fixed:opened}) and $V_{\sigma(i)} > \tau_{\sigma(i)}$ (Line \ref{strategy:fixed:accepted}). This holds if and only if $Y_{\sigma(i)} = \min\{r_{\sigma(i)}, V_{\sigma(i)}\} > \tau_{\sigma(i)}$ and $Y_{\sigma(j)} \leq V_{\sigma(j)} \leq  \tau_{\sigma(j)}$ for all $j$ such that $j< i$. This is the definition of the accepted box $i^*_{\sigma}$ for $\InstanceProphet$. The proof follows.
\end{proof}

Using the two lemmata, it is now easy to show our main result about the Pandora's Box With Value Discounting.

\begin{proof}[Proof of Theorem \ref{theorem:discounting-main}]
    Fix an arbitrary instance $\Instance$ of the Pandora's Box With Value Discounting. Let $\InstanceProphet = (Y_i)_{i \in [n]}$ be its induced Free-Order Prophet Inequality instance with $Y_i = \min\{r_i, V_i\}$ for $i \in [n]$. For $(\Instance, \FixedCostStrategy)$, the following holds, completing the proof:
        $\EX\left[u_{\Instance}(\FixedCostStrategy)\right] \geq \EX\left[Y_{i^*_{\sigma}}\right] \geq 0.7258 \cdot \EX\Big[\max_{i \in [n]} Y_i \Big] \geq 0.7258 \cdot \EX\left[u_{\Instance}(\pi^*)\right]$. 
    The first inequality follows from Lemma \ref{lemma:u-is-fo}, the second inequality follows from Lemma \ref{lemma:prophet-inequality}, and the last inequality follows from Lemma~\ref{lemma:opt-prophet-ub}. 
\end{proof}

\section{Conclusions}
The original Pandora's Box problem, proposed by \citet{Weitzman79}, included a specific discounting factor for costs and rewards. More recently, starting with the work of \citet{KleinbergWW16} and subsequent papers, this dependence on time has been omitted, with the exception of the works discussed in our Introduction. In this work, we introduce Pandora's Box Over Time, a class of online selection problems that generalizes the Pandora's Box problem (including the discounting proposed by \citet{Weitzman79}) to scenarios where time plays a richer role, capturing all the time-related extensions that have appeared in the literature. For this NP-Hard problem, we devise a constant-approximation strategy that exhibits a simple two-phase structure: first, it determines a non-adaptive inspection schedule and then computes an adaptive stopping rule for the scheduled boxes. This design template, which crucially exploits the adaptivity gap result of \citet{BradacSZ19}, is the predominant approach in the literature.

Whether it is possible to improve our approximation guarantees, potentially by going beyond the above design template, remains an open question for future research. Another interesting open question is whether we can rule out the existence of a PTAS for general Pandora's Box Over Time instances or any of the special cases we study in Section \ref{sec:special-cases}. Hardness of approximation results, such as that of \citet{BoodaghiansFLL20}, are rather scarce in this line of work.

\section*{Acknowledgements}
GA has been partially supported by project MIS 5154714 of the National Recovery and Resilience Plan Greece 2.0 funded by the European Union under the NextGenerationEU Program.

\medskip
\noindent
TE is supported by the Harvard University Center of Mathematical Sciences and Applications. 

\medskip
\noindent
MF has been partially funded by the European Research Council (ERC) under the European Union's Horizon 2020 research and innovation program (grant agreement No. 866132), by an Amazon Research Award, by the Israel Science Foundation Breakthrough Program (grant No.~2600/24), and by a grant from TAU Center for AI and Data Science (TAD).

\medskip
\noindent
FF has been partially supported by the FAIR (Future Artificial Intelligence Research) project
PE0000013, funded by the NextGenerationEU program within the PNRR-PE-AI scheme (M4C2, investment 1.3, line on Artificial Intelligence), by ERC Advanced Grant 788893 AMDROMA “Algorithmic and Mechanism Design Research in Online Markets”, and by PNRR MUR project IR0000013-SoBigData.it. 

\medskip
\noindent
AT was partially supported by the Gravitation Project NETWORKS, grant no. 024.002.003, and the EU Horizon 2020 Research and Innovation Program under the Marie Skłodowska-Curie Grant
Agreement, grant no. 101034253.

\appendix

\section{Missing Proofs}
\label{app:missing}

\subsection{Proof of Lemma \ref{lemma:kleinberg-lemma}}
\noindent For each $i \in [n]$, we write $I_i(\pi):= I_i$ and $A_i(\pi):= A_i$. We have:
\begin{align*}
     \EX\bigg[\bigg(\sum_{i=1}^n A_i V_i - I_i c_i \bigg)\bigg] &= \EX\bigg[\bigg(\sum_{i=1}^n A_i V_i - I_i \EX[(V_i - r_i)^+] \bigg)\bigg] \\
    &= \sum_{i=1}^n \Big(\EX\big[A_i V_i \big] - \EX\big[I_i (V_i - r_i)^+ \big] \Big)\\
    &= \sum_{i=1}^n \EX\big[A_i Y_i + A_i (V_i - r_i)^+ - I_i (V_i - r_i)^+ \big] \\
    &\leq \sum_{i=1}^n \EX\big[A_i Y_i + I_i (V_i - r_i)^+ - I_i (V_i - r_i)^+ \big] \\
    &= \sum_{i=1}^n \EX\big[A_i Y_i \big].
\end{align*}
The first equality follows from the definition of $r_i$. The second equality holds since, for $i \in [n]$, $V_i$ and $I_i$ are independent; whether a strategy opens box $i$ cannot be aﬀected by the reward in it since $D_1,\dots, D_n$ are independent. Then, the third equality follows from the definition of $(Y_i)_{i \in [n]}$. Finally, the inequality follows since $A_i \leq I_i$ holds for each box $i \in [n]$, i.e., if a reward is collected, then its box must have been already inspected.

For the second statement, notice that the above derivation is satisfied with equality when
\[
\EX\Big[(A_i - I_i)(V_i - r_i)^+ \Big] = 0 
\]
holds for each $i \in [n]$. Since, by definition, $A_i \leq I_i$, this holds if and only if the probability that $I_i = 1$, $A_i = 0$, and $V_i > r_i$ is $0$, and the lemma follows. \qed
\smallskip

\subsection{Proof of Lemma~\ref{lemma:f-monotone-submodular}}
 \noindent The non-negativity of the random variables directly implies that $f(S) \ge 0$, for any $S\subseteq [n]$, whereas the function is normalized by definition. Further, for every $S \subseteq T \subseteq [n]$, it holds that $f(S)= \,\EX\left[\max_{i \in S}Y_i\right] \leq \allowbreak \,\EX\left[\max_{i \in T}Y_i\right] =f(T)$, and therefore $f$ is monotone.
    
 It remains to show that $f$ is submodular. Fix $S, T \subseteq [n]$ and let $i^*=\argmax_{i \in S\cup T}Y_i$ (where ties are broken lexicographically, hence the use of `$=$' instead of `$\in$'). We have:
 \begin{align*} 
         f(S\cup T) &+ f(S \cap T) = \EX\Big[\max_{i \in S\cup T}Y_i\Big] + \EX\Big[\max_{i \in S\cap T}Y_i\Big]= \EX\Big[\max_{i \in S\cup T}Y_i + \max_{i \in S\cap T}Y_i  \Big]\\
         &=\Pr\left[i^* \in S \right] \cdot \EX\Big[\max_{i \in S\cup T}Y_i + \max_{i \in S\cap T}Y_i \mid i^* \in S\Big] \\
         &\quad +\Pr\left[i^* \in T \setminus S \right] \cdot \EX\Big[\max_{i \in S\cup T}Y_i + \max_{i \in S\cap T}Y_i \mid i^* \in T \setminus S\Big]\\
         &\leq \Pr\left[i^* \in S \right] \cdot \EX\Big[\max_{i \in S}Y_i + \max_{i \in T}Y_i \mid i^* \in S\Big] \\
         & \quad +\Pr\left[i^* \in T \setminus S \right] \cdot \EX\Big[\max_{i \in T}Y_i + \max_{i \in S}Y_i \mid i^* \in T \setminus S\Big]\\ 
        &=\EX\Big[\max_{i \in S}Y_i + \max_{i \in T}Y_i \Big]=\EX\Big[\max_{i \in S}Y_i\Big] + \EX\Big[\max_{i \in T}Y_i \Big]\\
        &=f(S)+f(T)\,, 
     \end{align*}
    which concludes the proof.
 \qed

\smallskip

\subsection{Proof of Lemma \ref{lemma:time-slots-equiv}}
\noindent The proof of the first statement is straightforward. Given a strategy $\pi$ (resp.~$\pi_{\textsc{sl}}$) for $\Instance$ (resp.~$\InstanceInterval$), we construct a strategy $\pi_{_{\textsc{sl}}}$ (resp.~$\pi$) for $\InstanceInterval$ (resp.~$\Instance$) that follows the exact same action at each round, i.e., $S(\pi) = S(\pi')$. Therefore, the two strategies have an identical action history and, by construction, achieve an identical utility for the decision maker.

To prove the second statement, let $\mathcal{Z}(\pi)$ be the set of random tuples of inspected boxes $(i_1, \dots, i_k)$ with $t_{i_j}(\pi) \in \Phi_{i_j}$ for $j \in \{1, \dots, k-1\}$ and $t_{i_k}(\pi) \not\in \Phi_{i_k}$, generated by $\pi$. Further, let $\pi'$ be a strategy for $\Instance$ that mimics $\pi$ for every sequence of boxes $S(\pi)$, except for those sequences with a \emph{prefix} in $(i_1, \dots, i_k) \in \mathcal{Z}(\pi)$. Instead, for each such sequence, $\pi'$ generates $(i_1, \dots, i_{k-1})$ and \emph{halts} in round $t_{i_k}(\pi)$ without inspecting any box in the interval $[t_{i_{k-1}}(\pi) + 1, t_{i_k}(\pi)]$. Denote by $\mathcal{E}$ this random event and by $\mathcal{E}^c$ its complement. When $\mathcal{E}$ occurs, we have:

\begin{align*}
    u_{\Instance}(\pi) &= \max_{i \in S(\pi)} \bar{v}_i\left(V_{it_i(\pi)}, T_{\pi} - t_i(\pi)\right) - \!\!\sum_{i \in S(\pi)} \!\bar{c}_i\left(t_i(\pi)\right) \\
    &\leq \!\max_{i \in (i_1, \dots, i_{k-1})} \!\!\bar{v}_i\left(V_{it_i(\pi)}, T_{\pi} - t_i(\pi)\right) + \max_{i \in [n]}\max_{t \in [H]} V_{it} - \bar{c}_i(t_{i_k})\\
    &\quad- \!\!\sum_{i \in (i_1, \dots, i_{k-1})} \!\!\bar{c}_i\left(t_i(\pi)\right) \\
    &= \!\max_{i \in (i_1, \dots, i_{k-1})} \!\!\bar{v}_i\left(V_{it_i(\pi)}, T_{\pi} - t_i(\pi)\right) - \!\!\!\sum_{i \in (i_1, \dots, i_{k-1})} \!\!\!\!\!\!c_i \, + \max_{i \in [n]} \max_{t \in [H]} V_{it} - M \\
    &= u_{\Instance}(\pi') + \max_{i \in [n]} \max_{t \in [H]} V_{it} - M. \numberthis\label{eq:when-E-happens}
\end{align*}

The first equality follows from \eqref{eq:utility}, noting that $p_i = 0$ for all $i \in [n]$ in the proxy instance $\Instance$, by construction. The inequality holds since 
\[\max_{i \in (i_1, \dots, i_{k-1})} \bar{v}_i\left(V_{it_i(\pi)}, T_{\pi} - t_i(\pi)\right) \geq 0 \]
and
\[
\max_{i \in S(\pi)} \bar{v}_i\left(V_{it_i(\pi)}, T_{\pi} - t_i(\pi)\right) \leq \max_{i \in S(\pi)} V_{it_i(\pi)} \leq \max_{i \in [n]} \max_{t \in [H]} V_{it}\,,
\]
and because $\bar{c}_i(t) \geq 0$ for all $i \in [n]$ and all $t \in [H]$. The second equality follows from the definition of $\Instance$ and the fact that $\mathcal{E}$ occurs. Finally, the last equality follows from the definition of $\pi'$. We conclude that
\begin{align*}
    \EX[u_{\Instance}(\pi)] &= \Pr[\mathcal{E}] \cdot \EX[u_{\Instance}(\pi) \mid \mathcal{E}] + \Pr[\mathcal{E}^c] \cdot \EX[u_{\Instance}(\pi) \mid \mathcal{E}^c] \\
    &= \Pr[\mathcal{E}] \cdot \EX[u_{\Instance}(\pi) \mid \mathcal{E}] + \Pr[\mathcal{E}^c] \cdot \EX[u_{\Instance}(\pi') \mid \mathcal{E}^c] \\
    &\leq \Pr[\mathcal{E}] \cdot \EX\Big[u_{\Instance}(\pi') + \max_{i \in [n]} \max_{t \in [H]} V_{it} - M \mid \mathcal{E}\Big] + \Pr[\mathcal{E}^c] \cdot \EX[u_{\Instance}(\pi') \mid \mathcal{E}^c] \\
    &= \EX[u_{\Instance}(\pi')] + \Pr[\mathcal{E}] \Big(\EX\Big[\max_{i \in [n]} \max_{t \in [H]} V_{it} \mid \mathcal{E}\Big] - \EX\Big[\max_{i \in [n]} \max_{t \in [H]} V_{it}\Big]\Big)\\
    &\leq \EX[u_{\Instance}(\pi')]\,.
\end{align*}
The second equality follows from the fact that $\pi'$ mimics all actions of $\pi$ under the event $\mathcal{E}^c$, by construction. The first inequality follows from \eqref{eq:when-E-happens}. The lemma follows.\qed

\section{Approximating Submodular Block Matching via Contention Resolution Schemes}\label{subsec:crs}

Recall that in the Submodular Block Matching problem, we are given a bipartite hypergraph with a specific structure: every hyperedge contains exactly one ``left'' node, and a sequence of consecutive ``right'' nodes. Our goal is to maximize a monotone submodular function over the hyperedges, while maintaining feasibility, i.e., no element of the bipartite ground set can belong to more than one hyperedge (for a formal definition, we refer to the main body).
By inspecting the problem, it is not hard to see that it is a monotone submodular maximization problem subject to a $p$-system\footnote{A down-closed feasibility set system is called a $p$-system if, for any pair of bases $B_1$ and $B_2$ (a basis is any maximal feasible set with respect to inclusion), it holds that $|B_i| \le p \cdot |B_j|$ for $i,j \in \{1,2\}.$} constraint, where $p = \min\{n, \max_{i \in [n]} p_i\}$. If we were to use a generic polynomial-time algorithm for this class of problems, it is known that we could not guarantee an approximation ratio better than $1/p$ \citep{BadanidiyuruV14}. 
Instead, we show in this section that there is a polynomial time $0.188$-approximation algorithm for the problem using the Contention Resolution Scheme (CRS) framework \citep{FeigeV06,FeigeV10,FeldmanNS11,ChekuriVZ14}. 
We need some additional preliminaries here. For the sake of presentation, these were omitted from Section \ref{sec:prelims}, as they are only needed in this subsection. 

A solution to Submodular Block Matching is a subset of hyperedges (which in this case must also be a matching). One can think of subsets of hyperedges as vectors in $\{0, 1\}^{|E|}$, living in $\mathbb{R}^{|E|}$. This way, it is possible to talk about fractional solutions, which in turn will be rounded through the CRS.
We will use $P$ to denote the convex hull of the vectors that correspond to feasible integral solutions, i.e., $P$ is the convex hull of all matchings. It is easy to see that $P$ is \textit{down-closed}, i.e., if $x \in P$ and $0\le y\le x$, then $y \in P$, and \textit{solvable}, i.e., linear functions can be maximized over $P$ in time polynomial in $|E|$.  
Next, consider a vector $x \in P$ 
(which in our case will be the output of the Measured Continuous Greedy \citep{FeldmanNS11} on the relaxation of the problem, i.e., maximizing the multilinear relaxation of the objective over $P$. 
As $x$ is typically fractional, the idea is to round each of its coordinates independently with probability equal to the value of that coordinate; for us, a hyperedge $e$ is chosen independently with probability $x_e$. This results to obtaining a random set $\mathtt{R}(x)\subseteq E$, which might still not be feasible. 
A CRS will randomly remove some hyperedges from $\mathtt{R}(x)$, so that we obtain a matching. Nevertheless, this cannot be done arbitrarily; we need the CRS to satisfy a number of properties that will allow the good approximation guarantees of $x$ to be transferred (in expectation) to the final rounded solution. Recall that, for $b\in \mathbb{R}$, $bP = \{bx \,|\, x\in P\}$ denotes the scaling of $P$ by $b$.

\begin{definition}
Let $b, c\in [0,1]$. A random function $\pi: bP \times \{0, 1\}^{|E|} \to \{0, 1\}^{|E|}$ (where we write $\pi_x(A)$ rather than $\pi(x, A)$) is a \emph{monotone $(b, c)$-balanced CRS} for $P$ if
\begin{itemize}
	\item for any $x\in bP$ and any $A\subseteq E$, it holds that $\pi_x(A)\subseteq A$ and it is feasible (i.e., $\pi_x(A)$ is a matching);
	\item for any $x\in bP$ and any $A\subseteq B\subseteq E$, it holds that $\Pr[e\in \pi_x(A)] \ge \Pr[e\in \pi_x(B)]$;
	\item for any $x\in bP$ and any $e\in E$, it holds that $\Pr[e\in \pi_x(\mathtt{R}(x))] \ge c\cdot x_e$. 
\end{itemize}
\end{definition}

The importance of monotone $(b, c)$-balanced CRSs stems from the fact that, when combined with a continuous algorithm with strong approximation guarantees, they result in approximation algorithms for maximizing the corresponding objectives over the integral
points of $P$. The next theorem follows from the work of \citet{FeldmanNS11} (full details in \citep{Feldman13}) and assumes that the fractional point $x\in bP$ is the output of the Measured Continuous Greedy algorithm; see also \citep{BuchbinderF18}. The \textit{density} of the polytope $P\subseteq [0, 1]^{|E|}$ is defined as $d(P) = \min_i \{b_i / \sum_{e\in E} a_{ie}\}$, where $\sum_{e\in E} a_{ie} x_e \le b_i$ is the $i$-th inequality constraint defining $P$ (excluding the inequalities $0 \le x_e \le 1$ for all $e \in E$). It is easy to see that $d(P) \in [0,1]$. 

\begin{theorem}[Follows from \citet{FeldmanNS11}]\label{thm:crs_framework}
Let $P\subseteq [0,1]^{|E|}$ be a solvable down-closed convex polytope with $d(P)<1-\varepsilon$ for some $\varepsilon>0$, and let $\pi$ be a monotone $(b, c)$-balanced CRS for $P$. Then, 
there is a 
$\nicefrac{e^b}{c(e^b - 1)}$-approximation
algorithm for maximizing a monotone submodular
function over the integral points of $P$ in polynomial time.
\end{theorem}

For our $P$, it is easy to bound  $d(P)$ away from $1$. Indeed, if $x_{ij}$ is the (relaxation of the) indicator variable of whether the hyperedge $e(i, j) = \{ b_i \} \cup \{ t_k \mid k = j, \dots, j + p_i \}$ is selected or not, the inequality constraints defining $P$ are $\sum_{j\in[H]} x_{ij} \le 1$, for all $i\in [n]$, and $\sum_{i\in[n]} \sum_{k = j}^{j+p_i} x_{ik} \le 1$, for all $j\in [H]$. Clearly, $d(P)\le 1/n$. Moreover, it is easy to obtain a monotone $(b, e^{-2b})$-balanced CRS for our $P$, for any $b\in[0,1]$; then, for $b = \ln{1.5}$, Theorem \ref{thm:crs_framework} would imply a $\nicefrac{27}4$-approximation algorithm for  Submodular Block Matching. Instead, in Lemma \ref{lemma:our_crs}, we are going to use the composition of two existing CRSs to obtain a stronger guarantee via the next known lemma.

\begin{lemma}[\citet{BuchbinderF18}]\label{lemma:crs_composition}
For $i\in[2]$, let $\pi_i$ be a monotone $(b, c_i)$-balanced CRS for a down-closed body $P_i$. Then, there is a
monotone $(b, c_1 c_2)$-balanced CRS $\pi$ for $P_1 \cap P_2$, which can be computed efficiently if $\pi_1$ and $\pi_2$ can be computed efficiently.
\end{lemma}

\begin{lemma}\label{lemma:our_crs}
    There is a monotone $\big(b, e^{-b}(1-e^{-b})/b\big)$-balanced CRS for the convex hull $P$ of all matchings of Submodular Block Matching and it can be computed efficiently.
\end{lemma}

\begin{proof}
We will express $P$ as the intersection of two down-closed polytopes $P_1$ and $P_2$ for which strong CRSs exist. Then we are going to compose those using Lemma \ref{lemma:crs_composition}. The two polytopes correspond to the two distinct types of inequalities that define $P$; thus, it is straightforward that they too are down-closed. That is 
\[P_1 = \Big\{ x\in [0, 1]^{|E|} \,\big|\,  \sum_{j\in[H]} x_{ij} \le 1, \text{ for all } i\in [n] \Big\}\]
and
\[P_2 = \Big\{ x\in [0, 1]^{|E|} \,\big|\,  \sum_{i\in[n]} \sum_{k = j}^{j+p_i} x_{ik} \le 1, \text{ for all } j\in [H] \Big\}\,.\]
Now, it is not hard to see that $P_1$ is the convex hull of all the (characteristic vectors of) independent sets of a very simple partition matroid where the $i$-th block is the singleton $\{b_i\}$ (i.e., the $i$-th vertex on the ``left'' side $L$) and the corresponding capacity is $1$. \citet{ChekuriVZ14} have shown that for the convex hull of the characteristic vectors of the independent sets of any matroid there is a monotone $\big(b, (1-e^{-b})/b\big)$-balanced CRS, so this is also the case for $P_1$ that can be computed efficiently. We call this CRS $\pi_1$. 

For $P_2$ it is not as straightforward to identify its inequalities with the constraints of a problem that is known to have a CRS. Note, however, that the hyperedges in our case have a very special structure. Each hyperedge, when restricted on the ``right'' side $R$, only contains consecutive vertices. That is, when one only looks at the restriction of the hyperedges on $R$ and the matching constraint, this can be interpreted as an instance where we are given a set of intervals on the discrete number line, and a subset $S$ of those is feasible if no two intervals in $S$ intersect. That is, $P_2$ can be seen as the convex hull of the characteristic vectors of the feasible sets of the \emph{Submodular Independent Set in Interval Graphs} problem, for which \citet{Feldman13} explicitly provided a monotone $\big(b, e^{-b}\big)$-balanced CRS; we call this CRS $\pi_2$. Since $P = P_1 \cap P_2$, applying Lemma \ref{lemma:crs_composition} completes the proof.
\end{proof}

Given our discussion above about the density of $P$, an immediate consequence of Theorem \ref{thm:crs_framework} and Lemma \ref{lemma:our_crs} is Theorem~\ref{cor:block_matching}, restated below. The exact factor is obtained by setting $b = 0.5227$.

\block*

\section{Approximating Stochastic Max Bipartite Matching}\label{app:PTAS}

We first need to define the \emph{Exact Bipartite Matching} problem: Given a bipartite graph $(U, V, E)$ with integral weights on its edges, i.e., each $e\in E$ has a weight $w_e \in \mathbb{N}$, and an integer $W$, return a matching $M$ of weight exactly $W$, i.e., $\sum_{e\in M} w_e = W$.

We call the variant that asks for a \emph{perfect} matching instead, the \emph{Exact Bipartite Perfect Matching} problem. It is known that Exact Bipartite Perfect Matching admits a randomized pseudopolynomial algorithm (as a special case of the XPB and XIB problems of \citet{CameriniGM92}), and it is not hard to see that this extends to Exact Bipartite Matching as well. 
 
\begin{theorem}[\citet{CameriniGM92}]\label{thm:rppPerfectMatching}
There is a randomized pseudopolynomial algorithm for the Exact Bipartite Perfect Matching problem.
\end{theorem}

\begin{corollary}\label{cor:rppMatching}
There is a randomized pseudopolynomial algorithm for the Exact Bipartite Matching problem.
\end{corollary}

\begin{proof}
Let $G$ be a bipartite graph $(U, V, E)$ such that each $e\in E$ has a weight $w_e \in \mathbb{N}$, and $W$ be a given integer. It is easy to construct in polynomial time a new bipartite graph $G'$, such that $G$ has a matching of weight exactly $W$ if and only if $G'$ has a perfect matching of weight exactly $W$.

Without loss of generality, assume that $|U| = n_1 \le n_2 = |V|$. We define the following sets of vertices: $A = \{ a_1, \ldots, a_{n_2 - n_1} \}$ (assuming $n_2 - n_1$, otherwise $A = \emptyset$), $B = \{ b_1, \ldots, b_{n_1} \}$, and $C = \{ c_1, \ldots, c_{n_1} \}$. Then $G'$ is the complete bipartite graph on $U \cup A \cup B$ and $V \cup C$. Any ``old'' edge $e$, i.e., any edge of the form $e = \{u, v\}$ such that $u \in U$, $v \in V$ and $\{u, v\} \in E$, has a weight $w_e$; any other edge has a weight of $0$.

Clearly, if $M$ is a matching of weight $W$ in $G$ then this can be extended to a matching $M'$ of weight $W$ in $G'$ just by adding edges of zero weight, and conversely, if $M'$ is a matching of weight $W$ in $G'$ then we can drop any edges of zero weight and the matching $M$ that consists of the corresponding edges of what remains is a mathching of weight $W$ in $G$.

So, given $G$, one can construct the graph $G'$ as above, run the randomized pseudopolynomial algorithm of Theorem \ref{thm:rppPerfectMatching} for $G'$ to get a matching $M'$ (or a `NO' answer), and turn this output to a matching $M$ for $G$ as above (or return a `NO' answer, respectively).
\end{proof}

Above, we defined the exact version of Bipartite Matching, but we essentially need to talk about the stochastic version of the problem where the objective is to maximize the expected maximum weight (rather than the sum of weights). More generally, suppose that for a problem $\mathcal{P}$ the goal is to select a subset of a ground set $[n]$, each element $i$ of which has an integral weight $w_i$, according to a given feasibility constraint $\mathcal{F}\subseteq 2^{[n]}$, so as to maximize a function of the weights of the selected subset.
Then $\textsc{exact-sum-}\mathcal{P}$ will denote the version of the problem where an additional integer $W$ is given, and one needs to select a feasible subset $S\subseteq [n]$, such that $\sum_{i\in S} w_i = W$, i.e., a set of total weight \emph{exactly} $W$. Now consider the stochastic variant of this setting, where each element $i$ of the ground set $[n]$ is the index of a distribution $X_i$, according to which the weight $w_i$ will be sampled. Let $\textsc{stochastic-max-}\mathcal{P}$ denote the version of the problem, where one needs to select a feasible subset $S\subseteq [n]$, such that $\EX[\max_{i\in S} w_i]$ is maximized. In the Bipartite Matching case, the latter would be the problem of finding a bipartite matching so that the expectation of the heaviest edge is maximized.

As mentioned by \citet{ChenHLLLL16}, the next theorem follows from repeating their Appendix C.2 for $\textsc{stochastic-max-}\mathcal{P}$ (rather than the $K$-MAX problem) and applying the idea of encoding the signature vector with a single integer, as in the proofs of Theorem 1 of \citet{LiD19} and Theorem 1.1 of \citet{LiY13}. 

\begin{theorem}[follows from \citet{ChenHLLLL16} and \citet{LiD19}]\label{thm:ptas}
Assume there exists a (randomized) pseudo-polynomial time exact algorithm for $\textsc{exact-sum-}\mathcal{P}$ and a polynomial time $O(1)$-approximation algorithm for $\textsc{stochastic-max-}\mathcal{P}$. Then, for any $\varepsilon > 0$,
there is a polynomial time $(1+\varepsilon)$-approximation algorithm for $\textsc{stochastic-max-}\mathcal{P}$.
\end{theorem}

Let $\Instance$ be an instance of Pandora Over Time With Instant Inspection. The main observation here is that now $\mathcal{H}(\Instance)$  is not a hypergraph anymore, but a bipartite graph. That is, in this case, Submodular Block Matching becomes the much better understood Submodular Matching on Bipartite Graphs problem with the objective of Lemma \ref{lemma:f-monotone-submodular}; note that this is the same as the \emph{Stochastic-Max-Bipartite Matching} problem.

\begin{theorem}[\citet{LeeSV10}]\label{thm:LeeSviridenko}
    For any fixed $\delta >0$, a polynomial time $(2+\delta)$-approximation algorithm exists for Submodular Matching on Bipartite Graphs.
\end{theorem}

Now Theorem \ref{thm:ptas} combined with Corollary \ref{cor:rppMatching} and Theorem \ref{thm:LeeSviridenko}, gives us Lemma \ref{lemma:exp-max-matching}, restated here.

\stochastic*

\bibliographystyle{plainnat}
\bibliography{bibliography}

\end{document}